\documentclass[onecolumn, draftcls, 11pt]{IEEEtran}
\usepackage{epsfig,latexsym,amssymb,amsmath,graphics}
\usepackage{graphicx}
\usepackage{subfigure}
\usepackage{amssymb,dsfont,amsthm}
\usepackage{cite,color}
\usepackage{url}
\usepackage{booktabs}
\usepackage{algorithm}
\usepackage{algorithmicx}
\usepackage{algpseudocode}
\usepackage{stfloats}
\usepackage{verbatim}
\usepackage{epstopdf}
\usepackage{diagbox}
\usepackage[justification=centering]{caption}

\definecolor{ColorRed}{rgb}{1,0,0}

\allowdisplaybreaks[4]

\newtheorem{proposition}{Proposition}
\newtheorem{theorem}{Theorem}

\begin{document}
\title{Channel Estimation and Signal Detection for NLOS Ultraviolet Scattering Communication with Space Division Multiple Access}

\author{Yubo Zhang, Yuchen Pan, Chen Gong, Beiyuan Liu and Zhengyuan Xu
	\thanks{This work was supported by National Natural Science Foundation of China (Grant No. 62171428),
		Key Research Program of Frontier Sciences of CAS (Grant No. QYZDY-SSW-JSC003),
		and the Fundamental Research Funds for the Central Universities.
		
		Yubo Zhang, Yuchen Pan, Chen Gong, and Zhengyuan Xu are with Key Laboratory of Wireless-Optical Communications, Chinese Academy of Sciences, University of Science and Technology of China, Hefei, Anhui 230027, China.
		Email: \{zyb170057,ycpan\}@mail.ustc.edu.cn; \{cgong821,xuzy\}@ustc.edu.cn.
		
		Beiyuan Liu is with the National Engineering Laboratory
		for Integrated Aero-Space-Ground-Ocean Big Data Application Technology,
		Northwestern Polytechnical University, Xi’an, Shaanxi 710072, China. 
		Email: lby@nwpu.edu.cn.
		}}

\maketitle

\begin{abstract}

We design a receiver assembling several photomultipliers (PMTs) as an array to increase the field of view (FOV) of the receiver and adapt to multiuser situation over None-line-of-sight (NLOS) ultraviolet (UV) channels. Channel estimation and signal detection have been investigated according to the space division characteristics of the structure. Firstly, we adopt the balanced structure on the pilot matrix, analyze the channel estimation mean square error (MSE), and optimize the structure parameters. Then, with the estimated parameters, an analytical threshold detection rule is proposed as a preliminary work of multiuser detection. The detection rule can be optimized by analyzing the separability of two users based on the Gaussian approximation of Poisson weighted sum. To assess the effect of imperfect estimation, the sensitivity analysis of channel estimation error on two-user signal detection is performed. Moreover, we propose a successive elimination method for on-off keying (OOK) modulated multiuser symbol detection based on the previous threshold detection rule. A closed-form upper bound on the detection error rate is calculated, which turns out to be a good approximation of that of multiuser maximum-likelihood (ML) detection. The proposed successive elimination method is twenty times faster than the ML detection with negligible detection error rate degradation.

\end{abstract}



\section{Introductions}
NLOS optical wireless scattering communication (OWC) is an effective supplement of conventional wireless communication when perfect alignment and no blockage between the transmitter and receiver cannot be guaranteed \cite{xu2008ultraviolet}. In fact, OWC is attracting more and more attention because of its potential large bandwidth and capacity \cite{elgala2011indoor, ding2009modeling}. An advantage of NLOS OWC is that it can adapt to the environment requiring radio silence or with strong electromagnetic interference. Moreover, considering the large attenuation in the atmosphere \cite{aung2019performance}, it's helpful to limit the communication in a certain area, if higher communication security is desirable.  

Due to high atmospheric attenuation of UV spectrum, a photon-counting receiver is adopted, where the received signal arrives in the form of discrete photoelectrons, yielding a Poisson distribution. In fact, UV communication and its channel characteristics have been demonstrated experimentally in \cite{xiao2011non, chen2014experimental, liao2015long, borah2021single}. As a stochastic process, the transmission path of each photon is modeled by Monte-Carlo method. The analytical approximation of Monte-Carlo method is investigated in \cite{xu2008analytical, wang2010non}, and the corresponding semi-analytical modeling has also been extensively studied \cite{sun2016closed, drost2013ultraviolet, wu2019modeling}. Works \cite{gong2016optical, zou2018characterization} have evaluated the signal detection performance of PMT based receiver. The capacity of multiple-input single-output (MISO) Poisson channel \cite{haas2003capacity} and multiple-input multiple-output (MIMO) Poisson channel \cite{chakraborty2008outage} are investigated. Work in \cite{lapidoth1998poisson} analyzes the capacity region under power limitations for multiple access Poisson channel. The achievable rates and signal detection of none-orthogonal multiple-access (NOMA) and code-division multiple-access (CDMA) multi-user Poisson channel are analyzed in \cite{wang2018signal}, while the transmission schemes of multiple access uplink channel are illustrated in \cite{chaaban2016capacity, chaaban2017capacity}. For Poisson broadcast channel, \cite{kim2016superposition, gong2018signal} shows the performance of superposition coding in the downlink channel, while the degraded Poisson broadcast channel is analyzed in \cite{sokolovsky2005attainable}. 

Channel estimation is crucial for the subsequent signal detection and performance analysis. Works in \cite{gong2015channel, gong2016optical} address the Poisson channel parameters estimation based on least-square (LS) criterion and evaluate the MSE. The channel estimation based on multiple channels correlation is proposed in \cite{liu2019correlation}, while blind and semi-blind estimation based on Poisson channel is studied in \cite{liu2022blind}. Moreover, existing works have explored the signal detection of scattering communication. The detection criterion analysis is given in \cite{chatzidiamantis2010generalized, gong2015lmmse}, while the works relevant to signal detection under relay transmission are shown in \cite{gong2018full, ardakani2017performance}. 

Up to now, works on space division ultraviolet multiuser system are in a quite limited number. An optimal beamforming design of indoor UV power-constrained multiuser communication with space division is proposed in \cite{arya2020state}. Works in \cite{ge2023shot} perform a joint design of precoder and equalizer for MIMO UV systems with pulse amplitude modulation (PAM). An interference cancellation aided pulse position modulation (PPM) scheme with spatial diversity and multiplexing over Poisson channels is proposed in \cite{zhou2012photon}. Works in \cite{abou2018spatial} consider the spatial-multiplexing solution to reach high bit rate under Poisson MIMO communication. As for signal detection, a suboptimal multiple-symbol detection method with equal gain combined statistics under Poisson MIMO channels is constructed in \cite{riediger2008multiple}, while \cite{gupta2012receiver} investigates the detection performance of zero-forcing (ZF), MMSE and maximum-likelihood sequence estimation (MLSE) receivers with spatial multiplexing over MIMO UV channels. To our very knowledge, none of the previous works systematically analyze the channel estimation and signal detection under OOK-modulated uplink-UV multiuser system. 

In fact, the idea of this paper is mainly illustrated by the issue of multiuser detection and seperation under a broad receiving FOV over UV uplink channels. PMT is selected as a detector because of its high detection sensitivity, but the FOV of a single PMT is quite narrow. Thus, a PMT array is adopted to meet the demand of broad FOV and its space division characteristics can be utilized. There are two main problems for ML detection under our system model. Firstly, the complexity of detection and performance analysis is too high with multiple users and receiving sectors. Secondly, ML detection rule is quite sensitive to the imperfect estimation of channel gain parameters. To circumvent these problems, a new detection method with lower complexity and tractable performance analysis is expected to be provided. Moreover, an accurate multiuser channel estimation scheme should be proposed to facilitate the signal detection. Under the slow fading assumption, we hope to accomplish the estimation with a pilot matrix before signal detection.





The contributions of this work can be summarized as follows. We adopt the balanced structure on the pilot matrix, analyze the channel estimation MSE, and optimize the structure parameters. We analyze the separability of two users based on the Gaussian approximation of Poisson weighted sum. An analytical threshold detection rule is proposed, and the linear weights as well as the threshold on the signal detection are optimized. The sensitivity analysis of channel estimation error on two-user signal detection is performed. Moreover, we propose a successive elimination method for OOK-modulated multiuser detection under inter-user interference (IUI) based on the previous threshold detection rule. A closed-form upper bound on the detection error probability is calculated, which turns out to be a good approximation of that of multiuser ML detection. The proposed successive elimination method is twenty times faster than the ML detection with negligible detection error degradation.

The remainder of this paper is organized as follows. In Section \uppercase\expandafter{\romannumeral2}, we introduce the space division receiver and construct an analytical signal model with LS channel estimation. In Section \uppercase\expandafter{\romannumeral3}, we propose the balanced pilot matrix design and analyze the channel estimation distortion. In Section \uppercase\expandafter{\romannumeral4}, we investigate the space division of two-user case, propose the optimal threshold detection method and analyze the sensitivity. In Section \uppercase\expandafter{\romannumeral5}, we propose a successive elimination signal detection approach with unknown IUI for the general multiuser detection, and provide a closed-form upper bound on the detection error probability. Numerical results are given in Section \uppercase\expandafter{\romannumeral6}. Finally, Section \uppercase\expandafter{\romannumeral7} concludes this work.

\def\degree{${}^{\circ}$}
\section{System Model} \label{sec.system_model}

\subsection{Model of Multiuser Space Division Receiver}
In the system, each user has one UV LED as a transmitter. Consider a UV communication network with multiple users and a receiver with multiple PMTs forming an array, where PMT detector is adopted due to its high detection sensitivity in the photon-counting regime. Fig.~\ref{node_struct} shows the three-dimensional structure of a PMT array. 

\begin{figure}[htbp]
	\centering
	\includegraphics[width=4.5in]{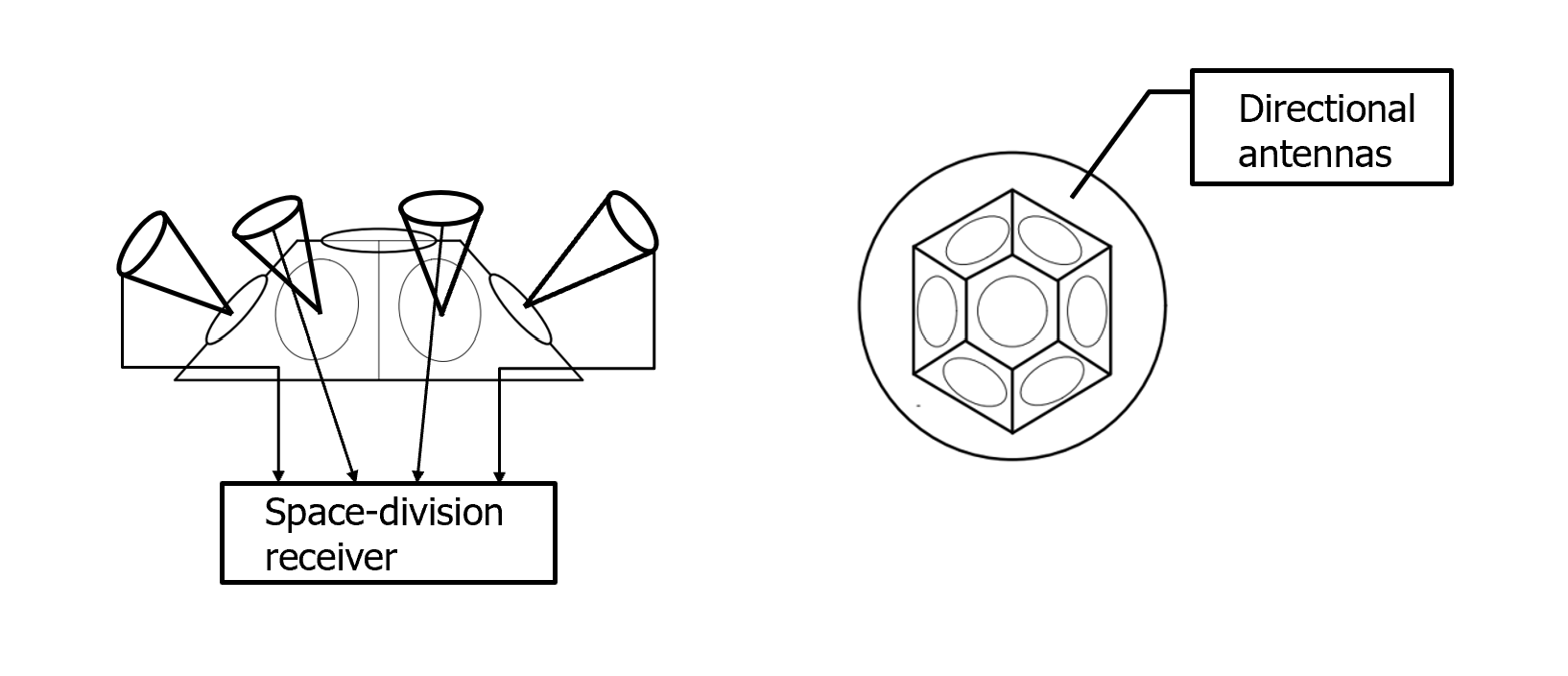}
	\caption{The 3D structure of a space division receiver.} \label{node_struct}
\end{figure}   

Assume that there are $K$ users and an array with $M$ PMTs as a space division receiver. To make the system model more comprehensible, Fig.~\ref{system_field} shows the planar graph of a space division multiuser uplink system with $K=3$ and $M=9$. 
\begin{figure}[htbp]
	\centering
	\includegraphics[width=7in]{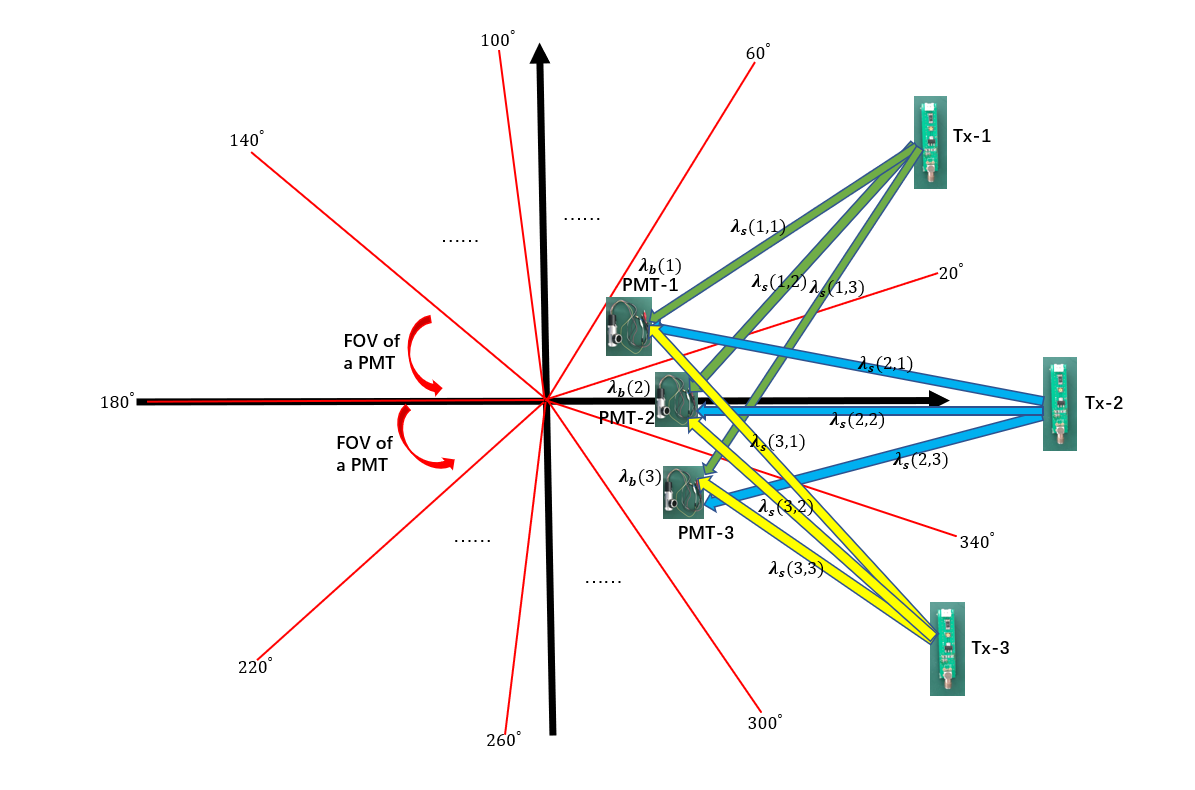}
	\caption{The system graph of a space division multiuser uplink system with $K=3$ and $M=9$ .} \label{system_field}
\end{figure} 

Due to extremely large path loss of NLOS UV link, the received signal can be characterized by discrete photoelectrons whose arrivals satisfy a Poisson random process. Denote the mean number of arrival photons of the signal from user $k$ to PMT $m$ as $\boldsymbol{\lambda_s}(k,m)$, and that of background noise of PMT $m$ as $\boldsymbol{\lambda_n}(m)$, respectively. Assume that all users are sending OOK modulated signals. Letting $N_{k,m}$ denote the number of received photoelectrons at PMT $m$ from user $k$, we have
\begin{gather} \label{basic_signal_model}
P(N_{k,m}=n|s_k=1)=\frac{(\boldsymbol{\lambda_s}(k,m)+\boldsymbol{\lambda_n}(m))^n}{n!}e^{-(\boldsymbol{\lambda_s}(k,m)+\boldsymbol{\lambda_n}(m))}, \notag \\
P(N_{k,m}=n|s_k=0)=\frac{ (\boldsymbol{\lambda_n}(m))^n}{n!}e^{-\boldsymbol{\lambda_n}(m)}, \notag \\
\boldsymbol{\lambda_s}(k,m)=\frac{\eta_{m}P_{k}T_{s}}{\xi_{k,m}h\nu}, \  \boldsymbol{\lambda_n}(m)=\boldsymbol{\Lambda_n}(m)T_s,
\end{gather}
where $\eta_{m}$ denotes the quantum efficiency of PMT $m$, $P_{k}$ denotes the emitting power of user $k$, $T_s$ denotes the symbol duration, and $\xi_{k,m}$ denotes the path loss from user $k$ to PMT $m$. 

The signals from multiple PMTs can be combined to improve the detection performance, expecially under multiuser interference. Assume that $(N_1,...,N_M)$ are the photon numbers of $M$ PMTs detected in a symbol slot. In this work, we adopt the following linear combination,
\begin{align} \label{statistics_propose}
W=\sum_{m=1}^{M} \alpha_{m}N_{m},
\end{align}
where $\{\alpha_{1},\alpha_{2},...,\alpha_{M}\}$ is the weight vector.

Based on Eq. (\ref{statistics_propose}), assuming independent photon arrivals at multiple receiving PMTs, the mean and variance of $W$ are given as follows,
\begin{align}  \label{statistics_miu_sigma}
\mathbb{E}(W)&=\sum_{m=1}^{M} \alpha_m\mathbb{E}(N_{m})=\sum_{m=1}^{M} \alpha_m\lambda_{m}, \notag \\
\mathbb{VAR}(W)&=\mathbb{E}(W^2)-\mathbb{E}^2(W)=\sum_{m=1}^{M} \alpha_m^2\lambda_{m}.
\end{align}

Note that weighted sum $W$ can be approximated via a Gaussian dustributed variable, denoted as $W^G \sim {\cal
N}(\mathbb{E}(W),\mathbb{VAR}(W))$. The accuracy of such approximation can be validated in Section VI.


\subsection{LMMSE-based Multiuser Channel Estimation}
In order to select the optimal signal detection method, channel parameters between all the users and receiving PMTs should be accurately estimated. Suppose that $K \times L$ estimation pilot matrix
\begin{align} \label{pilot_X_def}
\boldsymbol{X}=
\left(
\begin{array}{cccc}
x_{1,1} & x_{1,2} & ... & x_{1,L} \\
x_{2,1} & x_{2,2} & ... & x_{2,L} \\
... & ... & ... & ...\\
x_{K,1} & x_{K,2} & ... & x_{K,L}  
\end{array}
\right),
\end{align}
is adopted for channel estimaiton, which is known to the receiver. It means that each user will send a $L$-symbol length pilot sequence (For example, user $1$ sends $(x_{1,1},x_{1,2}, ...,x_{1,L})$). Assume that the channel parameters from all users to PMT $m$ are
\begin{align} \label{lam_sr_m}
\boldsymbol{\lambda_{sr}}(m)=[\boldsymbol{\lambda_s}(1,m),\boldsymbol{\lambda_s}(2,m),...,\boldsymbol{\lambda_s}(K,m)]^{T}.
\end{align}
All the parameters shown in Eq. (\ref{lam_sr_m}) can be estimated using the detected photon number of PMT $m$ in $L$ symbol durations, denoted as
\begin{align}
\boldsymbol{u_m}=[\boldsymbol{u_m}(1),\boldsymbol{u_m}(2),...,\boldsymbol{u_m}(L)]^{T}.
\end{align}
In fact, each element of $\boldsymbol{u_m}$ satisfies Poisson distribution. Noting that the Poisson arrival intensity of each symbol duration is influenced by $K$ users simultaneously, we have that
\begin{align} \label{pilot_data_distri}
&P(\boldsymbol{u_m}(l)=n)=\frac{(\boldsymbol{\lambda_m}(l))^n}{n!}e^{-(\boldsymbol{\lambda_m}(l))}, \notag \\
&\boldsymbol{\lambda_m}(l)=\sum_{k=1}^{K}\boldsymbol{\lambda_s}(k,m)x_{k,l}+\boldsymbol{\lambda_n}(m).
\end{align}

Similar to \cite{gong2015channel}, an unbiased channel estimator, denoted as $\boldsymbol{\hat{\lambda}_{sr}}(m)$, is given as follows,
\begin{align} \label{unbiased_H}
\boldsymbol{\hat{\lambda}_{sr}}(m)=(\boldsymbol{X}\boldsymbol{X}^T)^{-1}\boldsymbol{X}(\boldsymbol{u_m}-\boldsymbol{\lambda_n}(m)\boldsymbol{1}),
\end{align}
and the corresponding estimation MSE is given by,
\begin{align} \label{unbiased_H_MSE}
\mathbb{E} \Vert \boldsymbol{\lambda_{sr}}(m)-\boldsymbol{\hat{\lambda}_{sr}}(m) \Vert^2 
=Tr((\boldsymbol{X}\boldsymbol{X}^T)^{-1}\boldsymbol{X}diag(\boldsymbol{X}^T\boldsymbol{\lambda_{sr}}(m))\boldsymbol{X}^T(\boldsymbol{X}\boldsymbol{X}^T)^{-1})+\boldsymbol{\lambda_n}(m)Tr((\boldsymbol{X}\boldsymbol{X}^T)^{-1}).
\end{align}



\graphicspath{{figures/FigChanEsti/}}
\def\degree{${}^{\circ}$}

\section{Channel Estimation Design for Multiuser Space Division System} \label{sec.chan_esti}


\subsection{Pilot Matrix Design Based on Balanced Structure} 
In the system, all the users' signals are OOK modulated symbols. Let $\boldsymbol{X}=(\boldsymbol{x}_1,\boldsymbol{x}_2,...,\boldsymbol{x}_L)$, where each $\boldsymbol{x}_l$ is a $K \times 1$ colomn vector corresponding to the pilot symbols in slot $l$. 



Balanced structure is adopted for the design of pilot matrix $\boldsymbol{X}$, where each user sends with the same frequency in a block. In our design, it's expected that any $R$ users selected from the whole $K$ users simultaneously appear with the same frequency for any $k=1,2,...,K$. To explain this property more clearly, given matrix $\boldsymbol{X}$ as shown in Eq. (\ref{pilot_X_def}) and an index group $\Theta=\{\theta_1,...,\theta_R\}$, define   
\begin{align} \label{zeta_def}
\zeta_{\boldsymbol{X}}(\Theta) \triangleq \sum_{l=1}^{L} (\prod_{r=1}^{R} x_{\Theta(r),l}).
\end{align}
Based on Eq. (\ref{zeta_def}), for $\forall \Theta \subseteq \{1,2,...,K\}$ with $\Vert \Theta \Vert_0=R$ for any $1 \leq R \leq K$, it's expected that $\zeta_{\boldsymbol{X}}(\Theta)$ is only a function of $R$. Matrix $\boldsymbol{X}$ will be considered to be a balanced structure if such a property is satisfied.  

%

Considering $R$, define matrix
\begin{align} \label{mathx_w_def}
\boldsymbol{X}^{(R)}=(\boldsymbol{x}_1,\boldsymbol{x}_2,...,\boldsymbol{x}_{C^R_K}),
\end{align}
where $\boldsymbol{x}_i$ is a $K \times 1$ vector and $\Vert \boldsymbol{x}_i \Vert_1=R$ for $1 \leq i \leq C^R_K$. $C^R_K$ indicates the possible modes given weight $R$. An example is given as follows for $K=4$ and $R=2$, 
\begin{align} \label{x_weight}
\boldsymbol{X}^{(2)}=
\left(
\begin{array}{cccccc}
1 & 1 & 1 & 0 & 0 & 0 \\
1 & 0 & 0 & 1 & 1 & 0 \\
0 & 1 & 0 & 1 & 0 & 1 \\
0 & 0 & 1 & 0 & 1 & 1 
\end{array}
\right).
\end{align}

Based on the above definition, define group $\boldsymbol{\Xi}$ as the set of matrix generated via column concatenation of matrices in $\{\boldsymbol{X}^{(1)},\boldsymbol{X}^{(2)},...,\boldsymbol{X}^{(K)}\}$ with totally $2^K-1$ elements. Specifically, we let
\begin{align} \label{Xi_construct}
&\boldsymbol{\Xi}=\boldsymbol{\Xi}^1 \cup \boldsymbol{\Xi}^2 \cup ... \cup \boldsymbol{\Xi}^{K}, \notag \\
\boldsymbol{\Xi}^1&=\{\boldsymbol{X}^{(1)},\boldsymbol{X}^{(2)},...,\boldsymbol{X}^{(K)}\}, \notag \\
\boldsymbol{\Xi}^2=\{[\boldsymbol{X}^{(1)}, \boldsymbol{X}&^{(2)}],[\boldsymbol{X}^{(1)}, \boldsymbol{X}^{(3)}],...,[\boldsymbol{X}^{(K-1)}, \boldsymbol{X}^{(K)}]\},\notag \\
&...... \notag \\
\boldsymbol{\Xi}^{K-1}=\{[\boldsymbol{X}^{(2)},\boldsymbol{X}^{(3)},...,\boldsymbol{X}^{(K)}]&,[\boldsymbol{X}^{(1)},\boldsymbol{X}^{(3)},...,\boldsymbol{X}^{(K)}],...,[\boldsymbol{X}^{(1)},\boldsymbol{X}^{(2)},...,\boldsymbol{X}^{(K-1)}]\}, \notag \\
\boldsymbol{\Xi}^{K}&=\{[\boldsymbol{X}^{(1)}, \boldsymbol{X}^{(2)}, ...  \boldsymbol{X}^{(K)}]\}, \\
|\boldsymbol{\Xi}|=&\sum_{i=1}^{K}|\boldsymbol{\Xi}^{i}|=\sum_{i=1}^{K} C^i_{K}=2^{K}-1.
\end{align}



We have the following results on all matrices in $\boldsymbol{\Xi}$.
\begin{theorem} \label{single_fun}
	For $\forall \boldsymbol{X} \in \boldsymbol{\Xi}$ and $\forall \Theta \subseteq \{1,2,...,K\}$, while $\Vert \Theta \Vert_0=R$ for any $1 \leq R \leq K$, $\zeta_{\boldsymbol{X}}(\Theta)$ is only a function of $R$.
\end{theorem}

\begin{proof}
	For any index group $\boldsymbol{\beta}=\{\beta_1,...,\beta_S\} \subseteq \{1,2,...,K\}$, let
	\begin{align} \label{x_xi_def}
	\boldsymbol{X}^{\boldsymbol{\beta}}=(\boldsymbol{X}^{(\beta_1)};\boldsymbol{X}^{(\beta_2)};...;\boldsymbol{X}^{(\beta_S)}).
	\end{align}
	Each $\boldsymbol{X}^{(\beta_s)}$ is defined by Eq. (\ref{mathx_w_def}). Obviously, $\boldsymbol{X}^{\boldsymbol{\beta}} \in \boldsymbol{\Xi}$ and the colomn number of $\boldsymbol{X}^{\boldsymbol{\beta}}$ is $M(\boldsymbol{\beta})=\sum_{s=1}^{S} C^{\beta_s}_K$. Define
	\begin{align} 
	\binom{n}{k}=
	\begin{cases}
	&C^{k}_{n}, 0 \leq k \leq n, \notag \\
	&0, otherwise.
	\end{cases}
	\end{align}
	Considering any index group $\Theta \subseteq \{1,...,K\}$ with $\Vert \Theta \Vert_0=R$, according to Eq. (\ref{zeta_def}), we have
	\begin{align} \label{k_feasible_verify}
	\zeta_{\boldsymbol{X}^{\boldsymbol{\beta}}}(\Theta)=\sum_{l=1}^{M(\boldsymbol{\beta})} (\prod_{r=1}^{R} x_{\Theta(r),l}) 
	=\sum_{s=1}^{S} \ \sum_{h=1}^{C^{\boldsymbol{\beta}_i}_K} \ |(\prod_{r=1}^{R} x^{\beta_s}_{\Theta(r),h})| 
	=\sum_{s=1}^{S} \ \binom{K-R}{\beta_s-R}.
	\end{align}
	It can be seen that $\zeta_{\boldsymbol{X}^{\boldsymbol{\beta}}}(\Theta)$ is only a function of $R$ for $\forall R \in \{1,...,K\}$.
\end{proof}


\subsection{Objective Function Calculation From MSE Expression} 
Considering the estimation of $\boldsymbol{\lambda_{sr}}(m)=[\boldsymbol{\lambda_s}(1,m),\boldsymbol{\lambda_s}(2,m),...,\boldsymbol{\lambda_s}(K,m)]$ using the detected photon numbers of PMT $m$. Based on Eq. (\ref{unbiased_H_MSE}), assuming known background noise intensity $\boldsymbol{\lambda_n}(m)$. Denote $\Vert \boldsymbol{\lambda_{sr}}(m) \Vert_1$ as the element-wise first-order norm of $\boldsymbol{\lambda_{sr}}(m)$ as the channel gain parameters from $K$ users to PMT $m$. Considering pilot matrix $\boldsymbol{X}$ with balanced structure, we have the following preliminary results refer to Theorem \ref{single_fun}.
\begin{proposition} \label{abc_def}
	Consider pilot matrix
	\begin{align} \label{pilot_shown_2}
	\boldsymbol{X}=
	\left(
	\begin{array}{cccc}
	x_{1,1} & x_{1,2} & ... & x_{1,L} \\
	x_{2,1} & x_{2,2} & ... & x_{2,L} \\
	... & ... & ... & ...\\
	x_{K,1} & x_{K,2} & ... & x_{K,L}  
	\end{array}
	\right),
	\end{align}
	for $\forall i_1,i_2,i_3 \in \{1,2,...,K\}$ and $i_1 \neq i_2 \neq i_3$. It turns out that,
	\begin{align} \label{c_def}
	\frac{\sum_{l=1}^{L} x_{i,l}}{L} \triangleq a, \
	\frac{\sum_{l=1}^{L} x_{i_1,l} x_{i_2,l}}{L} \triangleq b, \
	\frac{\sum_{l=1}^{L} x_{i_1,l} x_{i_2,l} x_{i_3,l}}{L} \triangleq c,
	\end{align}
	are not related to specific values of $i_1$, $i_2$ and $i_3$.
\end{proposition}

Moreover, based on Eq. (\ref{k_feasible_verify}), if $ \boldsymbol{X}=\boldsymbol{X}^{\boldsymbol{\hat{\beta}}}$ for $\boldsymbol{\hat{\beta}}=\{\hat{\beta}_1,...,\hat{\beta}_S\}$, we have
\begin{align} \label{abc_calcu}
a=\frac{\sum_{s=1}^{S} \binom{K-1}{\hat{\beta}_s-1}}{\sum_{s=1}^{S} \binom{K}{\hat{\beta}_s}}, \
b=\frac{\sum_{s=1}^{S} \binom{K-2}{\hat{\beta}_s-2}}{\sum_{s=1}^{S} \binom{K}{\hat{\beta}_s}}, \
c=\frac{\sum_{s=1}^{S} \binom{K-3}{\hat{\beta}_s-3}}{\sum_{s=1}^{S} \binom{K}{\hat{\beta}_s}}.
\end{align} 

\begin{theorem} \label{f_mse_deep_calc}
	The estimation MSE of $\boldsymbol{\lambda_{sr}}(m)$ in Eq. (\ref{unbiased_H_MSE}), denoted as $\mathcal{F}_m(\boldsymbol{X})$, is given by 
	\begin{align} \label{f_mse_abc}
	&\mathcal{F}_m(\boldsymbol{X})=\frac{\boldsymbol{\lambda_n}(m)K(1-\frac{b}{a+(K-1)b})}{L(a-b)} \notag \\
	&+\frac{\Vert \boldsymbol{\lambda_{sr}}(m) \Vert_1}{L(a-b)^2}\bigg( a+(K-1)b+\frac{(2-K)b^2-2ab}{(a+(K-1)b)^2} (a+3b(K-1)+c(K-1)(K-2))\bigg),
	\end{align}
	where $a$, $b$ and $c$ are given by Eq. (\ref{abc_calcu}).
\end{theorem}

\begin{proof}
	See Appendix A.
\end{proof} 

The optimal pilot matrix design for the estimation of $\boldsymbol{\lambda_{sr}}(m)=[\boldsymbol{\lambda_s}(1,m),\boldsymbol{\lambda_s}(2,m),...,\boldsymbol{\lambda_s}(K,m)]$ is formulated by solving the following problem,
\begin{gather} \label{X_opt_problem}
\min\limits_{\boldsymbol{X}} \ \mathcal{F}_{m}(\boldsymbol{X}) \notag \\
s.t. \ \forall \Theta \subseteq \{1,2,...,K\}, \ \Vert \Theta \Vert_0=R, \ \zeta_{\boldsymbol{X}}(\Theta) \ is \ only \ a \ function \ of \ R.
\end{gather}

\graphicspath{{figures/FigDivFunc/}}
\def\degree{${}^{\circ}$}
\section{Analysis and Optimization of Threshold Detection: Two-User Case} \label{sec.div_function}
The multiuser signal detection will be performed after the channel estimation. Before investigating on the general multiuser situation, we consider binary decision under two-user situation as a crucial basis. Specifically, considering the problem of two-user separation via a PMT array, we propose signal detection with a linear weighted sum. Threshold detection rule is adopted. An optimization problem on the weight vector is proposed, and a tractable solution is obtained, based on which the detection threshold is derived. Furthermore, the closed-form error analysis on the proposed method will be given. Finally, the sensitivity analysis with respect to imperfect channel estimation is conducted.   

\subsection{Problem Formulation on Detection Error Minimization} \label{obj_derive}
Consider two users, user $A$ and user $B$, which are sending OOK modulated signal and each time only one user is allowed to send. We aim to design a detection rule to help the receiver decide which user is more likely to send signals. In this situation, the detection error rate equals the probability that the receiver chooses a wrong user. Assume that the mean photon arrival numbers of desired signal components from user $A$ to the $M$ PMTs are $(\boldsymbol{\lambda_a^s}(1),...,\boldsymbol{\lambda_a^s}(M))$, the mean photon arrival numbers of desired signal components from user $B$ to the $M$ PMTs are $(\boldsymbol{\lambda_b^s}(1),...,\boldsymbol{\lambda_b^s}(M))$, and the mean photon arrival numbers of background noise of $M$ PMTs are $\boldsymbol{\lambda_n}=(\boldsymbol{\lambda_n}(1),...,\boldsymbol{\lambda_n}(M))$. Denote that,
\begin{align}
\boldsymbol{\lambda_a}&(m)=\boldsymbol{\lambda_a^s}(m)+\boldsymbol{\lambda_n}(m), \ 
\boldsymbol{\lambda_b}(m)=\boldsymbol{\lambda_b^s}(m)+\boldsymbol{\lambda_n}(m), \ m=1,...,M. \\
&\boldsymbol{\lambda_{st}}(a)=(\boldsymbol{\lambda_a}(1),...,\boldsymbol{\lambda_a}(M)), \ \boldsymbol{\lambda_{st}}(b)=(\boldsymbol{\lambda_b}(1),...,\boldsymbol{\lambda_b}(M))
\end{align}
Assume that the pulse numbers of all the $M$ PMTs corresponding to user $A$ within a symbol duration are $(N_{a,1},...,N_{a,M})$, and those for user $B$ are $(N_{b,1},...,N_{b,M})$. Letting $\{\alpha_m\}^M_{m=1}$ denote the weights, the weighted sums are given by
\begin{align} \label{statis_display}
W_A=\sum_{m=1}^{M} \alpha_mN_{a,m}, \ W_B=\sum_{m=1}^{M} \alpha_mN_{b,m}.
\end{align}

Gaussian approximations of $W_A$ and $W_B$ are introduced, denoted as $W_A \sim \mathcal{N}(\mu_a,\sigma_a^2)$ and $W_B \sim \mathcal{N}(\mu_b,\sigma_b^2)$, respectively. According to Eq. (\ref{statistics_miu_sigma}), we have
\begin{align}  \label{SF_miu_sigma}
&\mu_a=\sum_{m=1}^{M} \alpha_m\boldsymbol{\lambda_a}(m), \ \mu_b=\sum_{m=1}^{M} \alpha_m\boldsymbol{\lambda_b}(m),\notag \\
&\sigma_a^2=\sum_{m=1}^{M} \alpha_m^2\boldsymbol{\lambda_a}(m), \ \sigma_b^2=\sum_{m=1}^{M} \alpha_m^2\boldsymbol{\lambda_b}(m).
\end{align}

Assume prior probability $p(a=1)=p(b=1)=0.5$. Since weight vector $\boldsymbol{\alpha}$ can be adjusted, without loss of generality, assume that $\mu_a \geq \mu_b$. Then, the detection error pobability with detection threshold $th$, denoted as $p_e$, is given by,
\begin{align} \label{detection_error_expression}
p_e=\frac{1}{2}\bigg(\int_{-\infty}^{th} \frac{1}{\sqrt{2\pi}\sigma_a} e^{-\frac{(x-\mu_a)^2}{2\sigma_a^2}}dx\bigg)+\frac{1}{2}\bigg(\int_{th}^{+\infty} \frac{1}{\sqrt{2\pi}\sigma_b} e^{-\frac{(x-\mu_b)^2}{2\sigma_b^2}}dx\bigg).
\end{align}
Based on Eq. (\ref{detection_error_expression}), the optimal detection threshold, denoted as $\hat{th}$, is given by
\begin{align} \label{threshold_expression}
\hat{th}=\frac{\sigma_a\sigma_b\sqrt{(\mu_a-\mu_b)^2+2(\sigma_a^2-\sigma_b^2)\ln{\frac{\sigma_a}{\sigma_b}}}+\sigma_a^2\mu_b-\sigma_b^2\mu_a}{\sigma_a^2-\sigma_b^2}.
\end{align}

Based on the above threshold, the detection error based on Gaussian approximation is given as follows,
\begin{align} \label{Pe_Q_calculation}
p_e = \frac{1}{2}\bigg(Q\bigg(\frac{\mu_a-\hat{th}}{\sigma_a}\bigg)+Q\bigg(\frac{\hat{th}-\mu_b}{\sigma_b}\bigg)\bigg),
\end{align}
where $Q(\bullet)$ denotes Gaussian Q-function. Let $V(A)=\frac{\mu_a-\hat{th}}{\sigma_a}$ and $V(B)=\frac{\hat{th}-\mu_b}{\sigma_b}$. We have the following
\begin{align} \label{V_S_1}
V(A)=\frac{\sigma_b(\mu_a-\mu_b)\sqrt{1+\frac{2(\sigma_a^2-\sigma_b^2)\ln{\frac{\sigma_a}{\sigma_b}}}{(\mu_a-\mu_b)^2}}+\sigma_a(\mu_b-\mu_a)}{\sigma_a^2-\sigma_b^2}.
\end{align}

Considering the scenario with large $\lambda_s$ and $\lambda_s/\lambda_n$, we have 
\begin{align} \label{xiaoliang_1}
[\frac{2(\sigma_a^2-\sigma_b^2)\ln{\frac{\sigma_a}{\sigma_b}}}{(\mu_a-\mu_b)^2}]^2 \ = \ O(\lambda_s^{-2}) \ \ll \ 1.
\end{align}
Since $\sqrt{1+X} \approx 1+\frac{1}{2}X$ when $X \ll 1$, we have
\begin{align} \label{V_S_2}
V(A) \approx \frac{(\sigma_b-\sigma_a)(\mu_a-\mu_b)+\frac{\sigma_b(\sigma_a^2-\sigma_b^2)\ln{\frac{\sigma_a}{\sigma_b}}}{\mu_a-\mu_b}}{\sigma_a^2-\sigma_b^2}=-\frac{\mu_a-\mu_b}{\sigma_a+\sigma_b}+\frac{\sigma_b\ln{\frac{\sigma_a}{\sigma_b}}}{\mu_a-\mu_b}.
\end{align}
Similar approximation holds for $V(F)$, given by 
\begin{align} \label{V_F_2}
V(B) &\approx -\frac{\mu_a-\mu_b}{\sigma_a+\sigma_b}+\frac{\sigma_b\ln{\frac{\sigma_a}{\sigma_b}}}{\mu_a-\mu_b}.
\end{align}

The detection error can be approximated as follows,
\begin{align} \label{Pe_final_expression}
p_e \approx Q\bigg(\frac{\mu_a-\mu_b}{\sigma_a+\sigma_b}\bigg).
\end{align}

We aim to maximize $\frac{\mu_a-\mu_b}{\sigma_a+\sigma_b}$, which can be written as follows according to Eq. (\ref{SF_miu_sigma}). 
\begin{align} \label{Div_expression_1}
Div(\boldsymbol{\alpha},\boldsymbol{\lambda_{st}}(a),\boldsymbol{\lambda_{st}}(b))=\frac{\sum_{m=1}^{M} \alpha_m(\boldsymbol{\lambda_a}(m)-\boldsymbol{\lambda_b}(m))}{\sqrt{\sum_{m=1}^{M} \alpha_m^2 \boldsymbol{\lambda_a}(m)}+\sqrt{\sum_{m=1}^{M} \alpha_m^2 \boldsymbol{\lambda_b}(m)}}.
\end{align}
Notice that Eq. (\ref{Div_expression_1}) is a homogeneous function to $\boldsymbol{\alpha}$. Without loss of generality, with constraint $\Vert \boldsymbol{\alpha} \Vert=1$, an optimization problem is proposed,
\begin{gather}
\max\limits_{\boldsymbol{\alpha}} \ Div(\boldsymbol{\alpha},\boldsymbol{\lambda_{st}}(a),\boldsymbol{\lambda_{st}}(b)), \notag \\
s.t \ \Vert \boldsymbol{\alpha} \Vert=1.
\end{gather}

\subsection{Optimal Design Based On Objective Function} 
Due to complicated form of the denominator, maximizing $Div(\boldsymbol{\alpha},\boldsymbol{\lambda_{st}}(a),\boldsymbol{\lambda_{st}}(b))$ is intractable. We resort to maximizing a lower bound. Since
\begin{align} \label{PMI}
\sqrt{\sum_{m=1}^{M} \alpha_m^2 \boldsymbol{\lambda_a}(m)}+\sqrt{\sum_{m=1}^{M} \alpha_m^2 \boldsymbol{\lambda_b}(m)} \leq \sqrt{2}\sqrt{\sum_{m=1}^{M}\alpha_m^2(\boldsymbol{\lambda_a}(m)+\boldsymbol{\lambda_b}(m))},
\end{align}
a lower bound on $Div(\boldsymbol{\alpha},\boldsymbol{\lambda_{st}}(a),\boldsymbol{\lambda_{st}}(b))$, denoted as $Div^{LB}(\boldsymbol{\alpha},\boldsymbol{\lambda_{st}}(a),\boldsymbol{\lambda_{st}}(b))$, is given by
\begin{align} \label{Div_expression_2}
Div^{LB}(\boldsymbol{\alpha},\boldsymbol{\lambda_{st}}(a),\boldsymbol{\lambda_{st}}(b))=\frac{\sum_{m=1}^{M} \alpha_m(\boldsymbol{\lambda_a}(m)-\boldsymbol{\lambda_b}(m))}{\sqrt{\sum_{m=1}^{M}\alpha_m^2(\boldsymbol{\lambda_a}(m)+\boldsymbol{\lambda_b}(m))}}.
\end{align}

The optimization problem can be transformed to maximizing the following lower bound,
\begin{gather} \label{div_optimal_problem}
\mathop{\max}_{\boldsymbol{\alpha}} \frac{(\sum_{m=1}^{M} \alpha_m(\boldsymbol{\lambda_a}(m)-\boldsymbol{\lambda_b}(m)))^2}{\sum_{m=1}^{M} \alpha_m^2(\boldsymbol{\lambda_a}(m)+\boldsymbol{\lambda_b}(m))}, \notag \\
s.t \ \Vert \boldsymbol{\alpha} \Vert=1.
\end{gather} 

Based on Cauchy-Inequality, we have that,
\begin{align} \label{cauchy_solve}
&(Div^{LB}(\boldsymbol{\alpha},\boldsymbol{\lambda_{st}}(a),\boldsymbol{\lambda_{st}}(b)))^2=\frac{(\sum_{m=1}^{M} \alpha_m(\boldsymbol{\lambda_a}(m)-\boldsymbol{\lambda_b}(m)))^2}{\sum_{m=1}^{M} \alpha_m^2(\boldsymbol{\lambda_a}(m)+\boldsymbol{\lambda_b}(m))} \notag \\
&=\frac{(\sum_{m=1}^{M} (\alpha_m\sqrt{\boldsymbol{\lambda_a}(m)+\boldsymbol{\lambda_b}(m)})(\frac{\boldsymbol{\lambda_a}(m)-\boldsymbol{\lambda_b}(m)}{\sqrt{\boldsymbol{\lambda_a}(m)+\boldsymbol{\lambda_b}(m)}}))^2}{\sum_{m=1}^{M} (\alpha_m\sqrt{\boldsymbol{\lambda_a}(m)+\boldsymbol{\lambda_b}(m)})^2} 
\leq \sum_{m=1}^{M} \frac{(\boldsymbol{\lambda_a}(m)-\boldsymbol{\lambda_b}(m))^2}{\boldsymbol{\lambda_a}(m)+\boldsymbol{\lambda_b}(m)}.
\end{align} 

The equality holds if and only if,
\begin{align} \label{cauchy_equal_condition}
\forall m, \ \frac{\alpha_m\sqrt{\boldsymbol{\lambda_a}(m)+\boldsymbol{\lambda_b}(m)}}{\frac{\boldsymbol{\lambda_a}(m)-\boldsymbol{\lambda_b}(m)}{\sqrt{\boldsymbol{\lambda_a}(m)+\boldsymbol{\lambda_b}(m)}}}=const,
\end{align}
and one solution to the optimal weight coefficients is given by
\begin{align}  \label{alpha_solution}
K_m=\frac{\boldsymbol{\lambda_a}(m)-\boldsymbol{\lambda_b}(m)}{\boldsymbol{\lambda_a}(m)+\boldsymbol{\lambda_b}(m)}, \ \hat{\alpha}_m=\frac{K_m}{\sqrt{\sum_{m^{'}=1}^{M} K_{m^{'}}^2}}, \ m=1,...,M.
\end{align}

Based on Eq. (\ref{threshold_expression}) and Eq. (\ref{alpha_solution}), an optimal detection rule aiming to seperate two users can be proposed. Denote an indicator function $I$, if user $A$ is detected, $I=1$, and if user $B$ is detected, $I=0$. Assume that the detected photon numbers are $(N_1,...,N_M)$. For certain threshold $\hat{th}$, we have that,
\begin{align} \label{two_user_detection}
\sum_{m=1}^{M} \hat{\alpha}_mN_m \ \mathop{\lessgtr}_{I=1}^{I=0} \ \hat{th}.
\end{align}


\subsection{Error Analysis Compared With ML Detection}
Given the channel parameter vectors of user $A$ and $B$, we can adopt Eq. (\ref{Pe_final_expression}) and Eq. (\ref{cauchy_solve}) to obtain the minimum detection error probability, denoted as $p_e^{th}(A,B)$, given by
\begin{align} \label{pe_sf_expression}
p_e^{th}(A,B)=p_e(\boldsymbol{\lambda_{st}}(a),\boldsymbol{\lambda_{st}}(b))
=Q\bigg(\sqrt{\frac{(\mu_a-\mu_b)^2}{2(\sigma_a^2+\sigma_b^2)}}\bigg)
=Q\bigg(\sqrt{\sum_{m=1}^{M} \frac{(\boldsymbol{\lambda_a}(m)-\boldsymbol{\lambda_b}(m))^2}{2(\boldsymbol{\lambda_a}(m)+\boldsymbol{\lambda_b}(m))}}\bigg).
\end{align} 

Suppose that the received photon numbers of $M$ PMTs are $(N_1,...,N_M)$. If ML detection is adopted to two-user seperation problem, the detection rule and the error calculation are shown as follows,
\begin{gather} \label{ml_single}
\frac{\prod_{m=1}^{M} \frac{(\boldsymbol{\lambda_a}(m))^{N_m}}{N_m!} e^{-\boldsymbol{\lambda_a}(m)}}{\prod_{m=1}^{M} \frac{(\boldsymbol{\lambda_b}(m))^{N_m}}{N_m!} e^{-\boldsymbol{\lambda_b}(m)}} \ \mathop{\lessgtr}_{A}^{B} \ \frac{p(b)}{p(a)}, \\
\label{ml_single_pe}
p_e^{ML}(A,B)=\sum_{N_1,...,N_M=0}^{+\infty} \frac{1}{2} min\bigg(\prod_{m=1}^{M} \frac{(\boldsymbol{\lambda_a}(m))^{N_m}}{N_m!} e^{-\boldsymbol{\lambda_a}(m)}, \prod_{m=1}^{M} \frac{(\boldsymbol{\lambda_b}(m))^{N_m}}{N_m!} e^{-\boldsymbol{\lambda_b}(m)}\bigg).
\end{gather}

Here we assume $p(a)=p(b)=\frac{1}{2}$. It can be seen that the calculation of $p_e^{ML}(A,B)$ is much more complicated than that of $p_e^{th}(A,B)$. Moreover, the expression of $p_e^{th}(A,B)$ is more tractable, and the simulation results show that $p_e^{th}(A,B)$ is pretty close to $p_e^{ML}(A,B)$. Thus, we can greatly simplify the performance evaluation procedure while maintaining the optimal behavior through threshold detection.

\subsection{Sensitivity Analysis of the Optimal Threshold Detection} \label{sensi_ana}
The analysis of these part is based on Eq. (\ref{pe_sf_expression}). Since there exists inevitable channel estimation errors, we aim to investigate the sensitivity of detection error probability to imperfect channel estimation. Let
\begin{align} \label{gsf_def}
G_{a,b}=\sqrt{\sum_{m=1}^{M} \frac{(\boldsymbol{\lambda_a}(m)-\boldsymbol{\lambda_b}(m))^2}{2(\boldsymbol{\lambda_a}(m)+\boldsymbol{\lambda_b}(m))}},
\end{align}
while $p_e$ can be expressed as follows,
\begin{align}
p_e=\int_{G_{a,b}}^{+\infty} \frac{1}{\sqrt{2\pi}} exp(-\frac{x^2}{2}) dt.
\end{align}
Assume that the estimation bias occurs on $\boldsymbol{\lambda_a}(m)$. Then, the absolute value of $\frac{\partial p_e}{\partial \boldsymbol{\lambda_a}(m)}$, which can characterize the detection sensitivity, is shown as follows,
\begin{align} \label{deriv_pe_lams}
\bigg| \frac{\partial p_e}{\partial \boldsymbol{\lambda_a}(m)} \bigg|=\frac{1}{\sqrt{2\pi}} exp(-\frac{{G_{a,b}}^2}{2}) \bigg|\frac{\partial G_{a,b}}{\partial \boldsymbol{\lambda_a}(m)} \bigg|
=\frac{1}{2\sqrt{2\pi}G_{a,b}} \ exp(-\frac{{G_{a,b}}^2}{2}) \ \bigg|1-\bigg(\frac{2}{1+\frac{\boldsymbol{\lambda_a}(m)}{\boldsymbol{\lambda_b}(m)}}\bigg)^2 \bigg|.  
\end{align}

Denote that $I_1=\frac{1}{2\sqrt{2\pi}G_{a,b}} \ exp(-\frac{{G_{a,b}}^2}{2})$ and $I_2=\bigg|1-\bigg(\frac{2}{1+\frac{\boldsymbol{\lambda_a}(m)}{\boldsymbol{\lambda_b}(m)}}\bigg)^2 \bigg|$. Letting $x=\frac{\boldsymbol{\lambda_a}(m)}{\boldsymbol{\lambda_b}(m)}$, obviously we have $x>0$. Thus, we have,
\begin{align} \label{I_b_f_x}
f(x)=1-\bigg(\frac{2}{1+x}\bigg)^2, \ I_2=|f(x)|, \ x>0.
\end{align}
Noting that $f(x)$ is monotonically increasing with respect to $x$, we have
\begin{align} \label{I_b_ub}
I_2=|f(x)| \leq \max\{\lim\limits_{x \rightarrow 0}|f(x)|, \lim\limits_{x \rightarrow \infty}|f(x)|\}=3.
\end{align}

In this way, an upper bound $I^{UB}_2=3$ is raised up to facilitate further calculation, implying that the detection error probability is not sensitive to imperfect estimation. Assume the following constraints on $\boldsymbol{\lambda_{st}}(a)$ and $\boldsymbol{\lambda_{st}}(b)$, we have
\begin{gather} \label{AB_constraint}
\sum_{m=1}^{M} (\boldsymbol{\lambda_a}(m)-\boldsymbol{\lambda_b}(m))^2 \geq C, \notag \\
\boldsymbol{\lambda_a}(m) \leq D, \ \boldsymbol{\lambda_b}(m) \leq D, m=1,...,M.
\end{gather}
According to Eq. (\ref{gsf_def}), a lower bound on $G_{a,b}$ can be calculated,
\begin{align} \label{G_AB}
G_{a,b} = \sqrt{\sum_{m=1}^{M} \frac{(\boldsymbol{\lambda_a}(m)-\boldsymbol{\lambda_b}(m))^2}{2(\boldsymbol{\lambda_a}(m)+\boldsymbol{\lambda_b}(m))}} \geq \sqrt{\frac{\sum_{m=1}^{M} (\boldsymbol{\lambda_a}(m)-\boldsymbol{\lambda_b}(m))^2}{4D}} \geq \sqrt{\frac{C}{4D}}.
\end{align}
Set that $G^{LB}=\sqrt{\frac{C}{4D}}$. Assume that the detection error probabilities with and without perfect channel estimation are $p_0$ and $p_1$, respectively. Assume that $\boldsymbol{\lambda_a}(m) \in (\lambda_{0}-\Delta, \lambda_{0}+\Delta)$, while $\lambda_{0}$ is the precise channel parameter and $\Delta$ is the maximum estimation bias. We have the following results on the detection sensitivity to channel estimation.

\begin{theorem} \label{sensi_analysis}
	For any $\Delta > 0$, there exists a constant $K=\frac{3exp(-\frac{C}{8D})}{\sqrt{2\pi\frac{C}{D}}}$, such that $|p_1-p_0| \leq K\Delta$ for $\forall \boldsymbol{\lambda_a}(m) \in (\lambda_{0}-\Delta, \lambda_{0}+\Delta)$. 
\end{theorem}

\begin{proof}
	Note that $p_1$ is a function of $\boldsymbol{\lambda_a}(m)$. Refer to the definition of integral, it's apparent that,
	\begin{align} \label{sensi_proof_1}
	|p_1-p_0| \ \leq \ \mathop{max}_{\boldsymbol{\lambda_a}(m)}\bigg(\bigg| \frac{\partial p_e}{\partial \boldsymbol{\lambda_a}(m)} \bigg|\bigg) \ |\boldsymbol{\lambda_a}(m)-\lambda_{0}| \ \leq \ \mathop{max}_{\boldsymbol{\lambda_a}(m)}\bigg(\bigg| \frac{\partial p_e}{\partial \boldsymbol{\lambda_a}(m)} \bigg|\bigg) \ \Delta.
	\end{align}
	Based on Eq. (\ref{deriv_pe_lams}) and Eq. (\ref{I_b_ub}), the remaining work is to find the upper bound on $I_1$. Noting that $I_1$ is decreasing with $G_{a,b}$, and a lower bound $G^{LB}$ can be obtained based on power constraint and space division constraint. According to Eq. (\ref{sensi_proof_1}), we have
	\begin{gather}
	|p_1-p_0| \ \leq \ I^{UB}_1 I^{UB}_2 \Delta
	= \frac{3}{2\sqrt{2\pi}G^{LB}} \ exp(-\frac{(G^{LB})^2}{2}) \ \Delta
	= \frac{3exp(-\frac{C}{8D})}{\sqrt{2\pi\frac{C}{D}}} \ \Delta.
	\end{gather}
	Thus, letting $K=\frac{3exp(-\frac{C}{8D})}{\sqrt{2\pi\frac{C}{D}}}$, the theorem is proved. 
\end{proof} 

%

The above theorem indicates that once the constraints are satisfied, the system performance degradeness caused by imperfect estimation will be bounded. Specifically, if two users are properly seperated and obey the transmission power constraint, the system performance is not sensitive to imperfect channel estimation.

\graphicspath{{figures/FigSigDetect/}}
\def\degree{${}^{\circ}$}

\section{Signal Detection Based On Space Division: General Multiuser Case}

In the previous Section, a preliminary threshold detection method aiming to seperate two users is proposed. In fact, ``two users'' can also be treated as ``two states''. Based on such idea and adopting the previous method, we can further investigate the signal detection method of general multiuser system. It turns out that similar to the method for two-user system, the proposed multiuser detection method also facilitate the performance analysis with IUI.

\subsection{Basic Settings}  
Assume that all channel link gains (channel parameters) are known to the receiver. For all $M$ PMTs, the detected pulse numbers within one symbol duration are given by $\boldsymbol{N}=(N_1,N_2,...,N_M)$. Assume that user $A$ is transmitting and its channel parameter vector is $\boldsymbol{\lambda_{st}}(a)=(\boldsymbol{\lambda_a}(1),\boldsymbol{\lambda_a}(2),...,\boldsymbol{\lambda_a}(M))$. We have that
\begin{align} \label{disting_N_distri}
	P(N_m=k)=e^{-\boldsymbol{\lambda_a}(m)} \frac{(\boldsymbol{\lambda_a}(m))^k}{k!}, \ \forall m \in \{1,2,...,M\}.
\end{align}  
In order to simplify the expression, denote that
\begin{align} \label{pn_lam_simple}
	p(\boldsymbol{N} | \boldsymbol{\lambda_{st}}(a))=\prod_{m=1}^{M} \frac{(\boldsymbol{\lambda_a}(m))^{N_m}}{N_m!} e^{-\boldsymbol{\lambda_a}(m)}.
\end{align}
For users $A$ and $B$, let
\begin{align} \label{pe_divide}
	P_e(A,B)=\frac{1}{2} (P_e(A|B)+P_e(B|A))
\end{align}
denote the detection error probability between the two users based on equal prior probability $p(A)=p(B)=\frac{1}{2}$. 

For the multiple users situation, assume that the $K$ interfering users are $(B_1,B_2,...,B_K)$ with channel parameter $\boldsymbol{\lambda_{st}}(b_k)=(\boldsymbol{\lambda_{b_k}}(1),\boldsymbol{\lambda_{b_k}}(2),...,\boldsymbol{\lambda_{b_k}}(M))$ for user $B_k$. Assume $p(a=0)=p(a=1)=\frac{1}{2}$ and $p(b_k=0)=p(b_k=1)=\frac{1}{2}$ for $k=1,...,K$. Moreover, denote the mean number of background noise for PMT array as $(\boldsymbol{\lambda_n}(1),...,\boldsymbol{\lambda_n}(M))$. Obviously there are totally $2^K$ combinations of $(b_1,...,b_K)$, denoted as modes $(\varphi_1,...,\varphi_{2^K})$. In this way, the channel parameters of interfering signal can be simplified as follows,
\begin{align} \label{chan_para_phi}
	\boldsymbol{\lambda_{\varphi^{*}}}(m)=\sum_{k=1}^{K} b_k^{*} \boldsymbol{\lambda_{b_k}}(m), \ m=1,...,M; \ \varphi^{*}=(b_1^{*},b_2^{*},...,b_{K}^{*}).
\end{align}

For user group $(A,B_1,B_2,...,B_{K})$, let $\boldsymbol{C}=(C_1,...,C_{2^K})$ and $\boldsymbol{D}=(D_1,...,D_{2^K})$ denote the states corresponding to $a=1$ and $a=0$, respectively, such that,
\begin{align} \label{T_arrange}
	\boldsymbol{\lambda_{C_{k}}}(m)=\boldsymbol{\lambda_a}(m)+\boldsymbol{\lambda_{\varphi_{k}}}(m)+\boldsymbol{\lambda_n}(m), \ &m=1,...,M, \ k=1,...,2^K; \notag \\
	\boldsymbol{\lambda_{D_{k}}}(m)=\boldsymbol{\lambda_{\varphi_{k}}}(m)+\boldsymbol{\lambda_n}(m), \ m=&1,...,M, \ k=1,...,2^K.
\end{align}



\subsection{ML Detection}
Firstly, we turn to ML detection as a baseline method for multiuser detection. Still for user group $(A,B_1,B_2,...,B_{K})$, based on Eq. (\ref{T_arrange}), the binary decision rule can be given as follows,
\begin{align} \label{ml_multi_rule}
\frac{\sum_{k=1}^{2^K} \prod_{m=1}^{M} \frac{(\boldsymbol{\lambda_{C_{k}}}(m))^{N_m}}{N_m!} e^{-\boldsymbol{\lambda_{C_{k}}}(m)}}{\sum_{k=1}^{2^K} \prod_{m=1}^{M} \frac{(\boldsymbol{\lambda_{D_{k}}}(m))^{N_m}}{N_m!} e^{-\boldsymbol{\lambda_{D_{k}}}(m)}} \ \mathop{\lessgtr}_{a=1}^{a=0} \ \frac{p(a=0)}{p(a=1)}.
\end{align}

Since $p(a=0)=p(a=1)=\frac{1}{2}$, the ML detection error probability can be calculated as follows,
\begin{align} \label{ml_pe}
P_e^{ML}(A,B_1,...,B_{K})=& \notag \\
\sum_{N_1,...,N_M=0}^{+\infty} \frac{1}{2^{K+1}} \cdot \min\bigg\{ \sum_{k=1}^{2^K} \prod_{m=1}^{M} \frac{(\boldsymbol{\lambda_{C_{k}}}(m))^{N_m}}{N_m!} e^{-\boldsymbol{\lambda_{C_{k}}}(m)}, &\sum_{k=1}^{2^K} \prod_{m=1}^{M} \frac{(\boldsymbol{\lambda_{D_{k}}}(m))^{N_m}}{N_m!} e^{-\boldsymbol{\lambda_{D_{k}}}(m)}\bigg\}.
\end{align}
The ML detection suffers a high computational complexity, especially when $M$ or $K$ becomes larger. It's expected to develop a new detection approach with lower real-time computational complexiy and tractable detection error analysis.

\subsection{Successive Elimination Method} \label{re}
A detection approach based on successive comparison and elimination is provided. The fundamental thought is that the signal detection with inter-user interference can be achieved by selecting the most likely state $T_{i_0}$ from $\boldsymbol{T}$ using threshold detection, leading to desirable symbol $\hat{a}$. Instead of precisely locating state $T_{i_0}$, we need to know whether $T_{i_0}$ belongs to set $\boldsymbol{C}$ or $\boldsymbol{D}$. Let
\begin{align} \label{u_ci_dj_N}
U(C_i,D_j,\boldsymbol{N})=\sum_{m=1}^{M} \alpha_m(\boldsymbol{\lambda_{C_i}}(m),\boldsymbol{\lambda_{D_j}}(m)) \cdot N_m-th(\boldsymbol{\lambda_{C_i}},\boldsymbol{\lambda_{D_j}}),
\end{align} 
where threshold $th(\boldsymbol{\lambda_{C_i}},\boldsymbol{\lambda_{D_j}})$ is given in Eq. (\ref{threshold_expression}). Based on the arguments in Section \ref{sec.div_function}, if $U(C_i,D_j,\boldsymbol{N})>0$, $C_i$ is more likely than $D_j$, and vice versa. Recall the error analysis in the previous sections, we have that,
\begin{align}
&P_e^{th}(C_i,D_j)=\frac{1}{2} (P_e^{th}(C_i|D_j)+P_e^{th}(D_j|C_i)) \notag \\
=\frac{1}{2} \bigg(p(&U(C_i,D_j,\boldsymbol{N})>0|D_j)+p(U(C_i,D_j,\boldsymbol{N})<0|C_{i})\bigg).
\end{align}
In this way, a multiuser detection method based on successive elimination tactic can be proposed. Based on two sets $\boldsymbol{C}=\{C_i\}_{i=1}^{2^K}$ and $\boldsymbol{D}=\{D_j\}_{j=1}^{2^K}$, the corresponding channel parameters can be given by Eq. (\ref{T_arrange}). With the detected photon numbers $\boldsymbol{N}$, the signal detection can be performed. During Round $1$, signal model $C_1$ is selected and compared with $(D_1,...,D_{2^K})$ sequentially. If $U(C_1,D_j,\boldsymbol{N})>0$, then the corresponding $D_j$ will be eliminated. As for $C_1$, if there exists at least one $D_j$ such that $U(C_1,D_j,\boldsymbol{N})<0$, $C_1$ will be eliminated when Round $1$ is over. The surviving members of $\boldsymbol{C}$ and $\boldsymbol{D}$ will enter Round $2$ and the same operations will be repeated. The procedure will be carried on until either $\boldsymbol{C}$ or $\boldsymbol{D}$ becomes empty. Obviously, if $\boldsymbol{C}$ finally becomes empty, detect $\hat{a}=0$. Otherwise, if $\boldsymbol{D}$ becomes empty, detect $\hat{a}=1$. 

%
%
%

Based on the proposed algorithm, we can propose a closed-form upper bound on the detection error probability, denoted as $P^{th}_e$, which can be calculated by the channel parameters aforementioned. 
\begin{theorem}
	Assume that $K$ is the total users number, $M$ is the total PMT number, $\boldsymbol{\lambda_a}(m)$ is the channel parameter between the desired user $A$ and $m$-th PMT, $\boldsymbol{\lambda_{\varphi_i}}(m)$ is the channel parameter summation (considering the additiveness of Poisson distribution) between interfering users $(B_1,...,B_K)$ and $m$-th PMT when interfering mode is $\varphi_i$. This explanation also adapts to $\boldsymbol{\lambda_{\varphi_j}}(m)$. $\boldsymbol{\lambda_n}(m)$ is the background noise intensity of the $m$-th PMT. We have that, 
	\begin{align} \label{pe_th_up_expression}
	P^{th}_e=\frac{1}{2^K} \bigg[ \sum_{i=1}^{2^{K}} \sum_{j=1}^{2^{K}} Q\bigg( \sqrt{\sum_{m=1}^{M} \frac{(\boldsymbol{\lambda_a}(m)+\boldsymbol{\lambda_{\varphi_i}}(m)-\boldsymbol{\lambda_{\varphi_j}}(m))^2}{2(\boldsymbol{\lambda_a}(m)+\boldsymbol{\lambda_{\varphi_i}}(m)+\boldsymbol{\lambda_{\varphi_j}}(m)+2\boldsymbol{\lambda_n}(m))}} \bigg) \bigg].
	\end{align}
\end{theorem}

\begin{proof}
	Denote $\mathcal{N}_1=\{\boldsymbol{N} \ | \ \hat{a}(\boldsymbol{N})=1\}$ and $\mathcal{N}_0=\{\boldsymbol{N} \ | \ \hat{a}(\boldsymbol{N})=0\}$. We have 
	\begin{align} \label{pe_re_first}
	P_e=&\frac{1}{2} \bigg(\sum_{\boldsymbol{N} \in \mathcal{N}_1} P(\boldsymbol{N}|a=0)+\sum_{\boldsymbol{N} \in \mathcal{N}_0} P(\boldsymbol{N}|a=1)\bigg) \notag \\
	=&\frac{1}{2^{K+1}} \bigg(\sum_{i=1}^{2^K} \sum_{\boldsymbol{N} \in \mathcal{N}_1} P(\boldsymbol{N}|\boldsymbol{\lambda_{D_i}}) + \sum_{i=1}^{2^K} \sum_{\boldsymbol{N} \in \mathcal{N}_0} P(\boldsymbol{N}|\boldsymbol{\lambda_{C_i}})\bigg).
	\end{align}
	
	For each $i$, we have,
	\begin{align} \label{right_set}
	\sum_{\boldsymbol{N} \in \mathcal{N}_0} P(\boldsymbol{N}|\boldsymbol{\lambda_{C_{i}}})=1-\sum_{\boldsymbol{N} \in \mathcal{N}_1} P(\boldsymbol{N}|\boldsymbol{\lambda_{C_{i}}}).
	\end{align}
	Construct
	\begin{align} \label{n_ci0_def}
	\mathcal{N}^{C_{i}}_1=\{\boldsymbol{N} \ | \ \mathop{\cap}_{j=1}^{2^K} [U(C_{i},D_j,\boldsymbol{N})>0]\}.
	\end{align}
	For all $\boldsymbol{N} \in \mathcal{N}^{C_{i}}_1$, we have $\hat{a}=1$, thus $\mathcal{N}^{C_{i}}_1 \subseteq \mathcal{N}_1$. Then, we have
	\begin{align} \label{pe_re_second}
	\sum_{\boldsymbol{N} \in \mathcal{N}_0} P(\boldsymbol{N}|\boldsymbol{\lambda_{C_{i}}})&=1-\sum_{\boldsymbol{N} \in \mathcal{N}_1} P(\boldsymbol{N}|\boldsymbol{\lambda_{C_{i}}}) \notag \\
	&\leq 1-\sum_{\boldsymbol{N} \in \mathcal{N}^{C_{i}}_1}  P(\boldsymbol{N}|\boldsymbol{\lambda_{C_{i}}}) \notag \\
	&=1-P(\mathop{\cap}_{j=1}^{2^K} [U(C_{i},D_j,\boldsymbol{N})>0|C_{i}]) \notag \\
	&=P(\mathop{\cup}_{j=1}^{2^K} [U(C_{i},D_j,\boldsymbol{N})<0|C_{i}]) \notag \\
	&\leq \sum_{j=1}^{2^K} P(U(C_{i},D_j,\boldsymbol{N})<0|C_{i})=\sum_{j=1}^{2^K} P_e^{th}(D_j |C_{i}).
	\end{align}
	
	Based on Eq. (\ref{pe_sf_expression}), Eq. (\ref{pe_divide}), Eq. (\ref{pe_re_first}) and Eq. (\ref{pe_re_second}), we have 
	\begin{align} \label{pe_th_multiuser}
	P_e &\leq \frac{1}{2^{K+1}} \bigg(\sum_{i=1}^{2^K} \sum_{j=1}^{2^K} P_e^{th}(C_i |D_j) + \sum_{i=1}^{2^K} \sum_{j=1}^{2^K} P_e^{th}(D_j |C_i)\bigg) \notag \\
	&=\frac{1}{2^K} \bigg( \sum_{i=1}^{2^{K}} \sum_{j=1}^{2^{K}} \frac{1}{2} (P_e^{th}(C_i |D_j)+P_e^{th}(D_j |C_i)) \bigg)  \notag \\
	&=\frac{1}{2^K} \bigg( \sum_{i=1}^{2^{K}} \sum_{j=1}^{2^{K}}  P_e^{th}(C_i,D_j)  \bigg) \notag \\
	&=\frac{1}{2^K} \bigg[ \sum_{i=1}^{2^{K}} \sum_{j=1}^{2^{K}} Q\bigg( \sqrt{\sum_{m=1}^{M} \frac{(\boldsymbol{\lambda_a}(m)+\boldsymbol{\lambda_{\varphi_i}}(m)-\boldsymbol{\lambda_{\varphi_j}}(m))^2}{2(\boldsymbol{\lambda_a}(m)+\boldsymbol{\lambda_{\varphi_i}}(m)+\boldsymbol{\lambda_{\varphi_j}}(m)+2\boldsymbol{\lambda_n}(m))}} \bigg) \bigg].
	\end{align}
\end{proof}


Since the weight vectors and detection threshold can be calculated in an offline manner, the successive elimination method shows lower real-time computational complexity, since $M$-order Poisson possibility multiplication is replaced by a linear weighted sum. 

\def\degree{${}^{\circ}$}



%

\section{Numerical Simulation}\label{simu}

\subsection{Gaussian Approximation of Poisson Weighted Sum}
We justify the approximation of Poisson weighted sum $W(\boldsymbol{\alpha},\boldsymbol{\lambda})$ using the corresponding Gaussian approximation $W^G(\boldsymbol{\alpha},\boldsymbol{\lambda})$. Assume that the space division receiver assembles three PMTs with a single user. Set the mean number of photons for the three sectors as $(\lambda_1,\lambda_2,\lambda_3)=(10,15,20)$ and the weight vector is randomly generated satisfying $\Vert \boldsymbol{\alpha} \Vert=1$. The PDF curves of $W(\boldsymbol{\alpha},\boldsymbol{\lambda})$ and $W^G(\boldsymbol{\alpha},\boldsymbol{\lambda})$ are depicted in Fig.~\ref{pdf_fit}, validating Gaussian approximation.
\begin{figure}[htbp]
	\centering
	\includegraphics[width=4in]{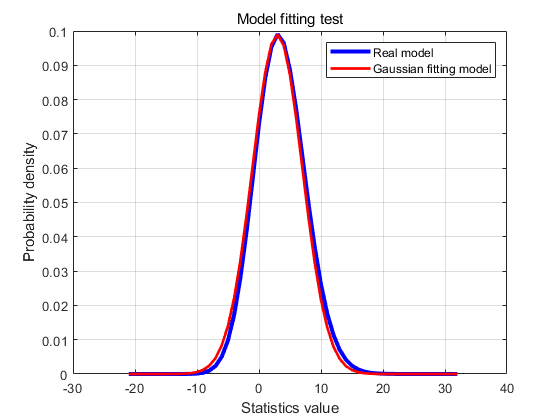}
	\caption{The probability density function of Poisson weighted sum and its Gaussian approxiamtion.} \label{pdf_fit}
\end{figure}




\subsection{Performance of Space Division Receiver: Two Users Seperation}
In this Subsection, the separability of a two-user system will be tested. The elevation angle of each transmitting LED and each receiving PMT is $30^{\circ}$, and each LED is straight to the receiver. Assume that the FOV angle of each PMT is $40^{\circ}$ and the full angular beamwidth of each LED is $60^{\circ}$. In this way, the 3D transceiver structure can be shown in Fig. \ref{trans_struct}. Recall that the detailed division of azimuth angle field is shown in Fig.~\ref{system_field}, while the radial distance between each user and the receiver is $100m$. Monte-Carlo simulations show that the Poisson-arriving signal from a single user can be depicted by three dominated PMTs, thus leading to inevitable inter-user interference. 
\begin{figure}[htbp]
	\centering
	\includegraphics[width=6in]{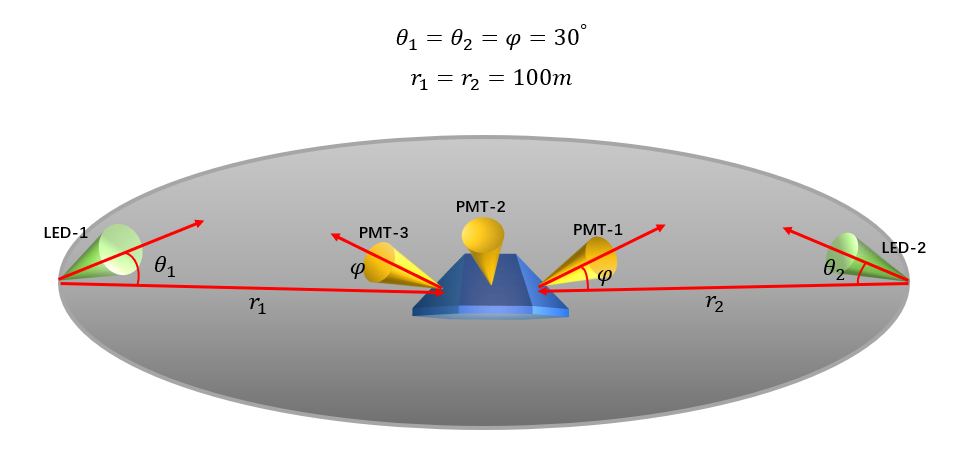}
	\caption{The 3D space division transceiver structure.} \label{trans_struct}
\end{figure} 

The detection error rate with respect to different angles of $A$ and $B$ is simulated. The performances are analyzed from two different perspectives. Firstly, set the angle difference between $A$ and $B$, which is denoted as $\Delta\theta(A,B)$, to be 5$^{\circ}$, 10$^{\circ}$ and 20$^{\circ}$ ($\theta(B)=\theta(A)+\Delta\theta(A,B)$), and change $\theta(A)$ from 0$^{\circ}$ to 120$^{\circ}$ (FOV of three PMTs) with step 5$^{\circ}$. Fig.~\ref{DER_pos_change} depicts the detection error rates under different $\theta(A)$. Both the radical distance of user $A$ ($r(A)$) and user $B$ ($r(B)$) are set to be $100m$. It shows that a larger $\Delta\theta(A,B)$ leads to a lower detection error rate. Moreover, there exists an obvious fluctuation when changing $\theta(A)$ from 0$^{\circ}$ to 120$^{\circ}$, indicating that the error rate is lower when $A$ locates at the edge of a FOV according to Fig.~\ref{system_field}. To justify this phenomenon more accurately, we change $\theta(A)$ from 0$^{\circ}$ to 40$^{\circ}$ with step 1$^{\circ}$, and fix $\Delta\theta(A,B)=8^{\circ}$. The simulated detection error rate is shown in Fig.~\ref{DER_single_sector}, proving that the detection error rate with $A$ locating at the edge of a FOV is lower than that with $A$ locating at the center of a FOV.   
 
%

\begin{figure}[htbp]
	\centering
	\begin{minipage}[t]{0.48\textwidth}  
		\centering
		\includegraphics[width=8.5cm]{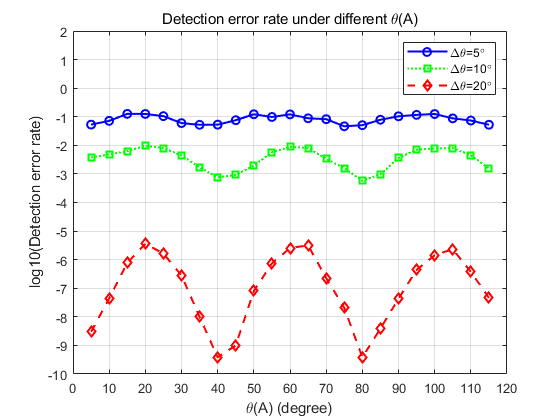}
		\caption{The detection error rate under different $\theta(A)$.}
		\label{DER_pos_change}
	\end{minipage}
	\begin{minipage}[t]{0.48\textwidth}  
		\centering
		\includegraphics[width=8.5cm]{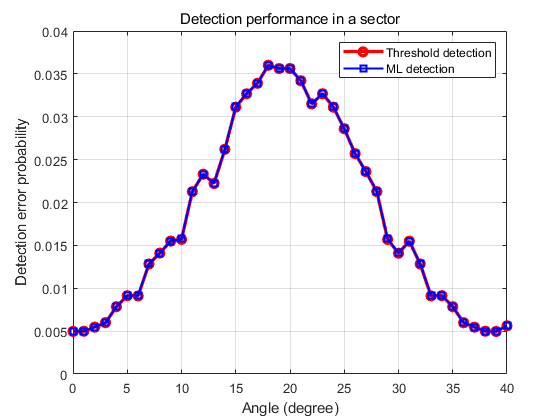}
		\caption{The detection error rate under different $\theta(A)$ within a single PMT FOV.}
		\label{DER_single_sector}
	\end{minipage}
\end{figure}


%

Then, to depict the relationship between detection error rate and $\Delta\theta(A,B)$ in detail, set the angle of $A$ ($\theta(A)$) to be 0$^{\circ}$, 20$^{\circ}$, 40$^{\circ}$ and 60$^{\circ}$, and change $\Delta\theta(A,B)$ from 3$^{\circ}$ to 20$^{\circ}$. Fig.~\ref{DER_diff_change} depicts the detection error rates under different $\Delta\theta(A,B)$. It can be observed that the detection error rate will decrease as angle difference $\Delta\theta(A,B)$ increases.

%

\begin{figure}[htbp]
	\centering
	\begin{minipage}[t]{0.48\textwidth}  
		\centering
		\includegraphics[width=8.5cm]{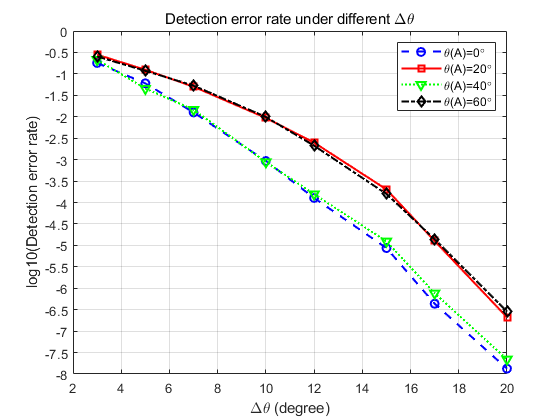}
		\caption{The detection error rate under different $\Delta\theta(A,B)$.}
		\label{DER_diff_change}
	\end{minipage}
	\begin{minipage}[t]{0.48\textwidth}  
		\centering
		\includegraphics[width=8.5cm]{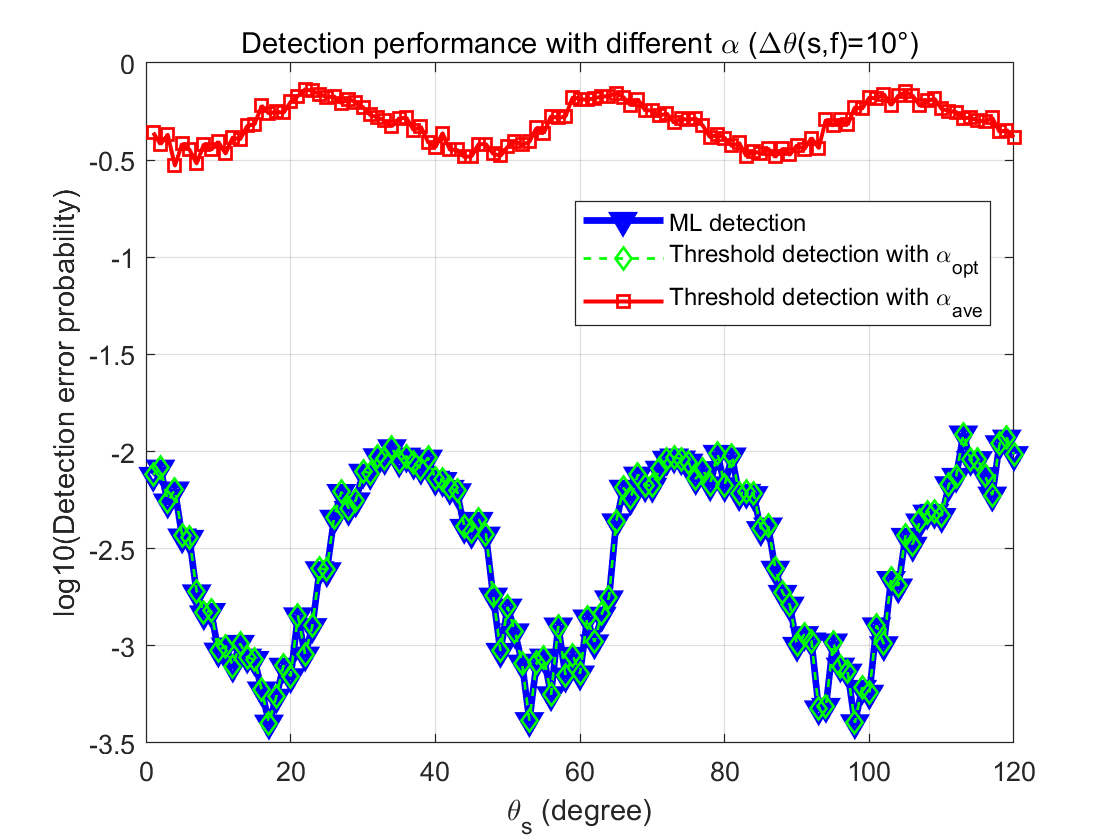}
		\caption{The comparison of two different threshold detections and ML detection.}
		\label{opt_ave_result}
	\end{minipage}
\end{figure}

%



Finally, we evaluate the performance of the two-user detection method proposed in Section \uppercase\expandafter{\romannumeral3}. Fixing $\Delta\theta(A,B)=10^{\circ}$ and changing $\theta(A)$ from 0$^{\circ}$ to $120^{\circ}$ with step $1^{\circ}$, we compare the detection error rate of the proposed detection method with that of maximum-likelihood (ML) detection method. For comparison, we also show the threshold detection with uniform weight ($\alpha_1=...=\alpha_M$). As shown in Fig.~\ref{opt_ave_result}, the performance of proposed threshold detection method is close to that of ML detection, and significantly outperforms that of uniform weight, justifying the asymptotical optimality of the proposed method. 



\subsection{Performance of Channel Estimation Based On Pilot Matrix}
In this Subsection, the accuracy of channel estimation based on pilot matrix designed in Section \ref{sec.chan_esti} is demonstrated. Set $K=4$ as the user number, and $3$ dominating PMTs providing detected photon numbers. As for users' azimuth angles, set $\theta_1=30^{\circ}, \theta_2=40^{\circ}, \theta_3=50^{\circ}$, and $\theta_4=60^{\circ}$. Monte-Carlo simulation shows that the photon arrival intensity between each user and PMT1 are $(1.0491,3.2533,9.6285,20.8329)$. The parameters for PMT2 and PMT3 are $(9.7798,3.1585,37.3374,22.3473)$ and $(37.1711,43.1114,1.0340,1.0000)$, respectively. 

Considering the overhead problem, usually the length of pilot sequence should not be longer than several hundreds. In the first part simulation, the pilot sequence length varies from $100$ to $500$ with step $100$. Moreover, for $K=4$, there are $14$ types of basic patterns of pilot matrix ($\boldsymbol{X}^{(4)}$ is singular, thus excluded) according to Eq. (\ref{Xi_construct}) (The definition of $\boldsymbol{X}^{(R)}$ is shown in Eq. (\ref{x_weight})). To verify the estimation accuracy, the theoretical estimation MSE from Eq. (\ref{f_mse_abc}) and the test MSE (Monte-Carlo simulation repeated 500 times) under each basic pattern are compared. Simulation shows that the basic pattern shown in Eq. (\ref{x_opt}), denoted as $\boldsymbol{X}_{opt}$, achieves the lowest MSE. To verify its optimality, the estimation MSE under two other patterns will be depicted, too. As shown in Fig.~\ref{chan_esti_pmt_1}, Fig.~\ref{chan_esti_pmt_2} and Fig.~\ref{chan_esti_pmt_3}, during the channel estimation of PMT1, PMT2 and PMT3, the estimation MSE under pattern $\boldsymbol{X}_{opt}$ is always lower than those under other two patterns with the same pilot length. Moreover, fixed the basic pattern, the theoretical MSE is quite close to the test MSE and both of them decrease as the length of pilot sequence increases, indicating that the MSE calculation shown in Eq. (\ref{f_mse_abc}) is feasible, and longer pilot sequence benefits the channel estimation. 
\begin{align} \label{x_opt}
\boldsymbol{X}_{opt}=
\left(
\begin{array}{ccccc}
1 & 0 & 0 & 0 & 1 \\
0 & 1 & 0 & 0 & 1 \\
0 & 0 & 1 & 0 & 1 \\
0 & 0 & 0 & 1 & 1  
\end{array}
\right).
\end{align}

\begin{figure}[t]
	\centering
	\subfigure[The MSE of channel parameters detected by PMT1.]{
		\label{chan_esti_pmt_1} 
		\includegraphics[scale=0.25]{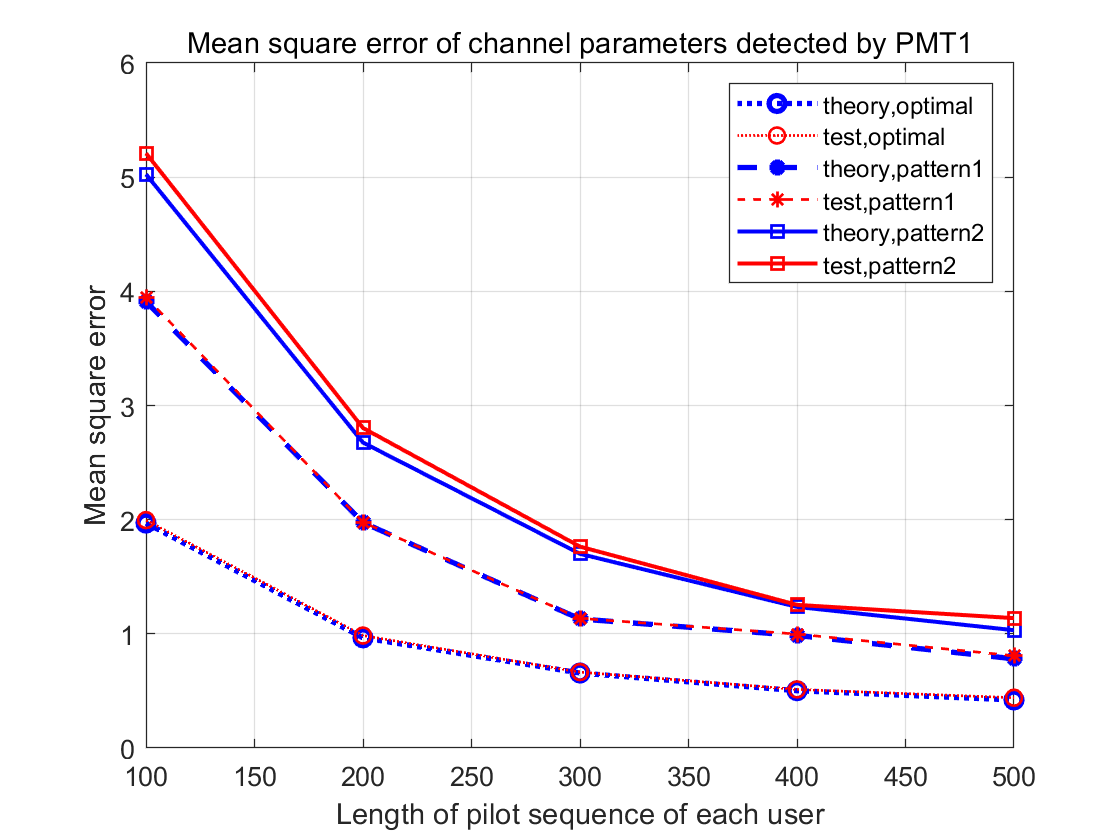}}
	\hspace{0.05in} 
	\subfigure[The MSE of channel parameters detected by PMT2.]{
		\label{chan_esti_pmt_2} 
		\includegraphics[scale=0.25]{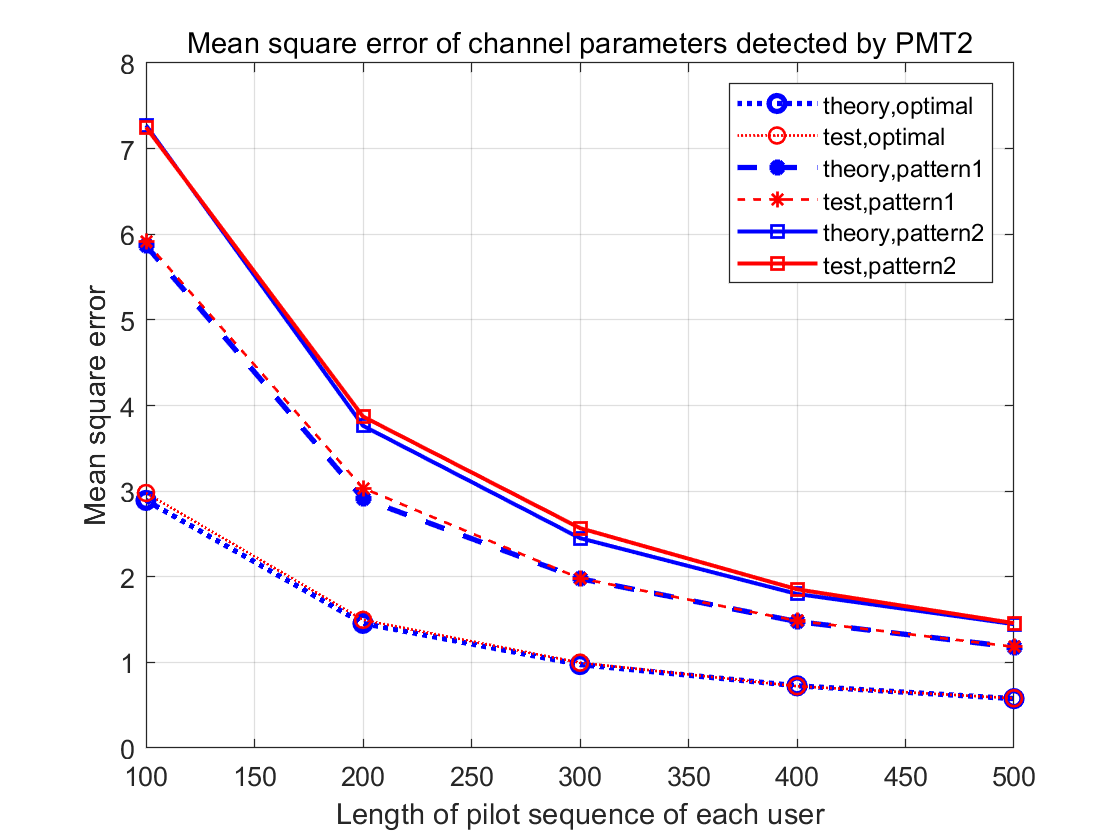}}
	\hspace{0.05in} 
	\subfigure[The MSE of channel parameters detected by PMT3.]{
		\label{chan_esti_pmt_3} 
		\includegraphics[scale=0.25]{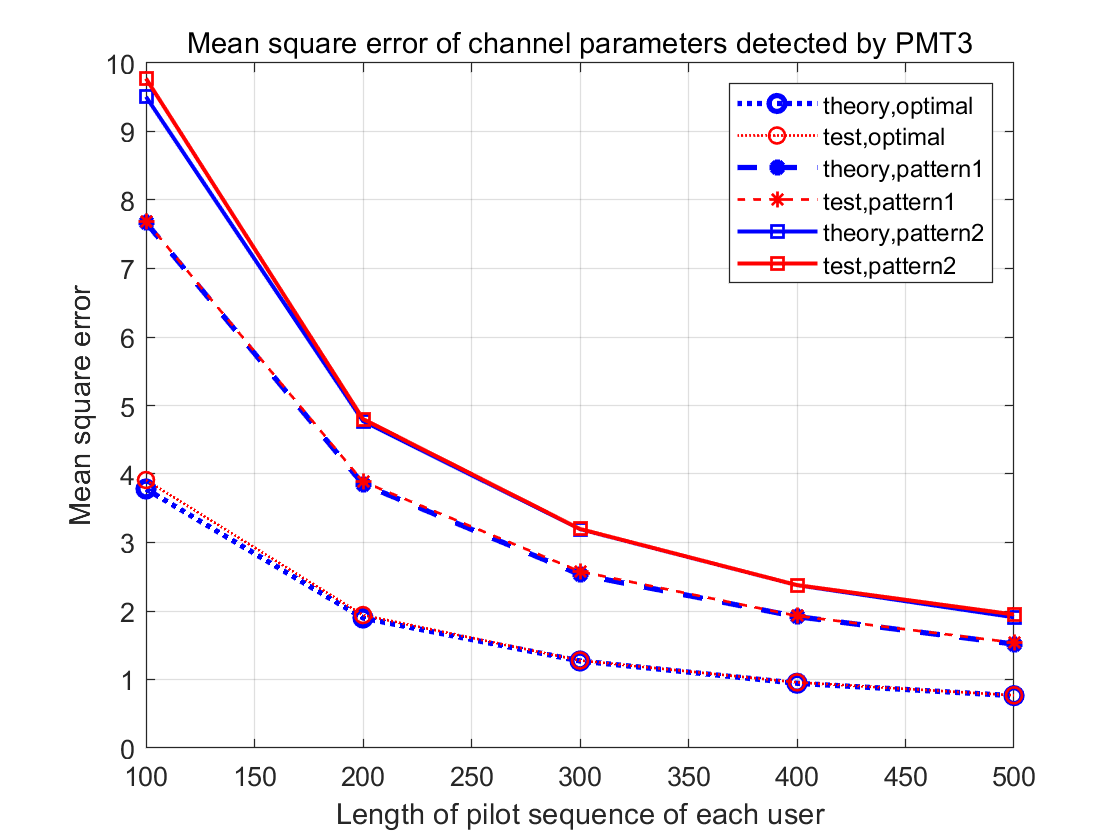}}
	\caption{The MSE of pilot matrix channel estimation.}
	\label{chan_esti_pmt} 
\end{figure}

Moreover, under basic settings aforementioned, we justify that the threshold detection error rate is not sensitive to imperfect channel estimation, which conforms to the conclusions in Subsection \ref{sensi_ana}. 


Firstly, we perform channel estimation under all the $14$ basic patterns of pilot matrix ($K=4, L=100$) according to Eq. (\ref{Xi_construct}), and perform the two-users division detection between user $1$ and user $2$. The basic pattern block is repeated multiple times until exceeding length $100$, and then columns in the last block are removed randomly to reduce the length to $100$ if necessary. We repeat the channel estimation for $500$ times, and the estimated parameters are adopted to signal detection simulation (with a $10^8$-length PRBS OOK test sequence) and the detection error rate is obtained. Compared with the real channel parameters, we calculate the estimation MSE, the mean absolute error (MAE) of estimation, and the mean symbol detection error rate (SER) under each pilot matrix pattern. Denote $\boldsymbol{X}^{i_1,...,i_k}=[\boldsymbol{X}^{(i_1)},...,\boldsymbol{X}^{(i_k)}]$, while $\boldsymbol{X}^{(i)}$ is shown in Eq. (\ref{x_weight}). For instance, $\boldsymbol{X}^{1,2,4}$ can be written as follows,
\begin{align}
\boldsymbol{X}^{1,2,4}=\left(
\begin{array}{ccccccccccc}
1 & 0 & 0 & 0 & 1 & 1 & 1 & 0 & 0 & 0 & 1\\
0 & 1 & 0 & 0 & 1 & 0 & 0 & 1 & 1 & 0 & 1 \\
0 & 0 & 1 & 0 & 0 & 1 & 0 & 1 & 0 & 1 & 1 \\
0 & 0 & 0 & 1 & 0 & 0 & 1 & 0 & 1 & 1 & 1 
\end{array}
\right).
\end{align}

The MSE, MAE and SER results are shown in Table \ref{tab1}. As a comparison, the SER with perfect channel estimation is $0.0522466$. It's seen that with the balanced basic pattern design of pilot matrix, the normalized estimation MSE/MAE is around $10^{-2} \sim 10^{-1}$, leading to only about $10^{-4} \sim 10^{-5}$ increment of the symbol detection error. Thus, under the balanced pilot matrix design, the detection performance is not sensitive to imperfect channel estimation. 

\begin{table}[h]
	\centering
	\caption{Channel estimation and signal detection performance under different balanced basic pattern of pilot matrix (K=4, L=100).}\label{tab1}
	\begin{tabular}{|l|c|c|c|}
		\hline
		\diagbox{Pattern}{Index} & Normalized MSE & Normalized MAE & SER \\\hline
		$\boldsymbol{X}^{1,4}$ & 0.0658 & 0.0539 & 0.0522515  \\\hline
		$\boldsymbol{X}^{2}$ & 0.0726 & 0.0593 & 0.0522547  \\\hline
		$\boldsymbol{X}^{1,2}$ & 0.0738 & 0.0589 & 0.0522615  \\\hline
		$\boldsymbol{X}^{1}$ & 0.0764 & 0.0607 & 0.0522687  \\\hline
		$\boldsymbol{X}^{1,2,4}$ & 0.0817 & 0.0623 & 0.0522721  \\\hline
		$\boldsymbol{X}^{1,2,3}$ & 0.0819 & 0.0636 & 0.0523509  \\\hline
		$\boldsymbol{X}^{2,4}$ & 0.0892 & 0.0660 & 0.0523644  \\\hline
		$\boldsymbol{X}^{1,3}$ & 0.0986 & 0.0691 & 0.0524071  \\\hline
		$\boldsymbol{X}^{1,2,3,4}$ & 0.0988 & 0.0701 & 0.0524266  \\\hline
		$\boldsymbol{X}^{2,3}$ & 0.0999 & 0.0707 & 0.0524789  \\\hline
		$\boldsymbol{X}^{2,3,4}$ & 0.1028 & 0.0717 & 0.0525042  \\\hline
		$\boldsymbol{X}^{1,3,4}$ & 0.1072 & 0.0726 & 0.0525576  \\\hline
		$\boldsymbol{X}^{3}$ & 0.1430 & 0.0838 & 0.0525860  \\\hline
		$\boldsymbol{X}^{3,4}$ & 0.1778 & 0.0923 & 0.0526244  \\\hline
	\end{tabular}
\end{table}

However, a large performance degradation may occur if a random pilot matrix design is chosen. Specifically, when the length of pilot matrix ($L$) is short (usually $L \textless 50$), an ill-conditioned pilot matrix, which can lead to large estimation error and detection error, is more likely to be constructed. To justify this point, under each length $L$, we repeat random pilot matrix design for $1000$ times and obtain the worst detection SER and the corresponding MSE/MAE, and the performance of the optimal balanced pilot matrix ($\boldsymbol{X}^{1,4}$) is also shown. From Table \ref{tab2}, it is seen that shorter $L$ leads to a worse detection performance of random pilot matrix design, but that of the optimal balanced pilot matrix design only has a little change ($10^{-3}$). The performance difference between these two designs will decrease with a larger $L$. In conclusion, when choosing a short pilot matrix aiming to decrease the overhead, we will risk suffering a high estimation error and a high detection error rate if adopting the random pilot matrix design. Thus, it's beneficial to choose the optimal balanced pilot matrix design.      

\begin{table}[h]
	\centering
	\caption{Channel estimation and signal detection performance under the optimal balanced pilot matrix and the random designed pilot matrix (K=4).}\label{tab2}
	\begin{tabular}{|l|c|c|c|c|c|c|}
		\hline
		\diagbox{$L$}{Index} & MSE(random) & MAE(random) & SER(random) &  MSE(opt) & MAE(opt) & SER(opt)\\\hline
		$20$ & 4.0252 & 8.9659 & 0.1195937 & 0.3261 & 0.2447 & 0.0533203  \\\hline
		$30$ & 3.4445 & 5.6216 & 0.0925499 & 0.2203 & 0.1269 & 0.0528763  \\\hline
		$50$ & 2.2872 & 4.9198 & 0.0791195 & 0.1369 & 0.0826 & 0.0525223  \\\hline
		$70$ & 1.7522 & 2.6499 & 0.0622262 & 0.0915 & 0.0682 & 0.0523786 \\\hline
		$100$ & 0.9241 & 1.7515 & 0.0593964 & 0.0682 & 0.0576 & 0.0522471  \\\hline
	\end{tabular}
\end{table}

\subsection{Performance of Multiuser Signal Detection with IUI}
We focus on the multiuser threshold detection performance with unknown IUI, and compare it to that of ML detection. The signal detection with 1,2 and 3 interfering users is simulated. Different angle pairs will be set to observe how the detection performance will change as users' positions change. 


For the scenario with two users, set $r(A)=100m, r(B)=100m$. We aim to test the detection error rate with respect to different angle difference $\Delta\theta(A,B)$. Recall that we adopt $9$ sectors within the whole angle field. Fix $\theta(A)=0^{\circ}$ and $\theta(A)=20^{\circ}$, and change $\Delta\theta(A,B)$ from $2^{\circ}$ to $30^{\circ}$. Fig.~\ref{s0_interf_1} and Fig.~\ref{s20_interf_1} depict the detection error rate under different $\Delta\theta(A,B)$ for $\theta(A)=0^{\circ}$ and $\theta(A)=20^{\circ}$, respectively. It's apparent that a larger angle difference between user $A$ and user $B$ leads to a lower detection error rate. By comparing the two figures, we find that the desired user ($A$) located at the edge of a PMT FOV has a better detection performance than that at the center of a PMT FOV, which conforms to the previous conclusion. Moreover, it's seen that three lines, representing $P_e^{ML}$, $P_e^{test}$ and $P_e^{th}$, are quite close to each other while maintaining $P_e^{ML} \textless P_e^{test} \textless P_e^{th}$. It indicates that $P_e^{th}$ is indeed an excellent closed-form approximation of $P_e^{ML}$. In addition, the successive elimination method is proved to be feasible since $P_e^{test}$ is also close to $P_e^{ML}$.

\begin{figure}[H]
	\centering
	\subfigure[The detection error rate with one interfering user when $\theta(A)=0^{\circ}$.]{
		\label{s0_interf_1} 
		\includegraphics[scale=0.25]{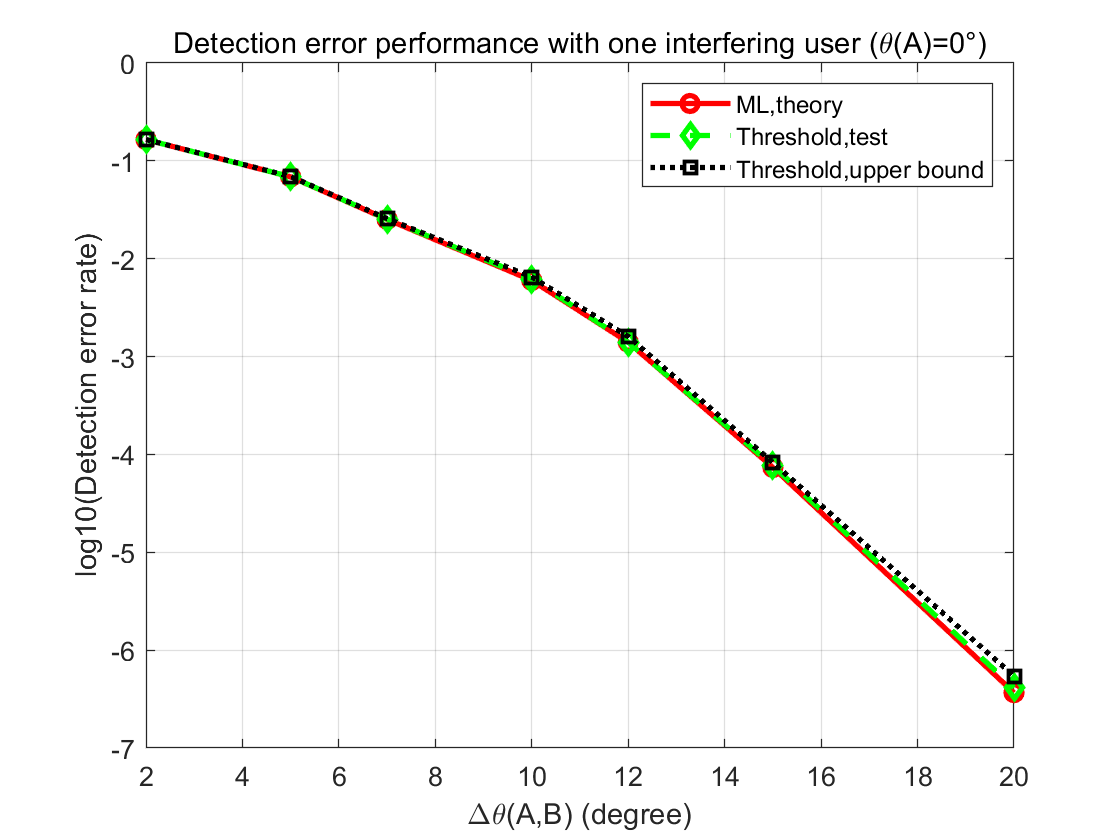}}
	\hspace{0.05in} 
	\subfigure[The detection error rate with one interfering user when $\theta(A)=20^{\circ}$.]{
		\label{s20_interf_1} 
		\includegraphics[scale=0.25]{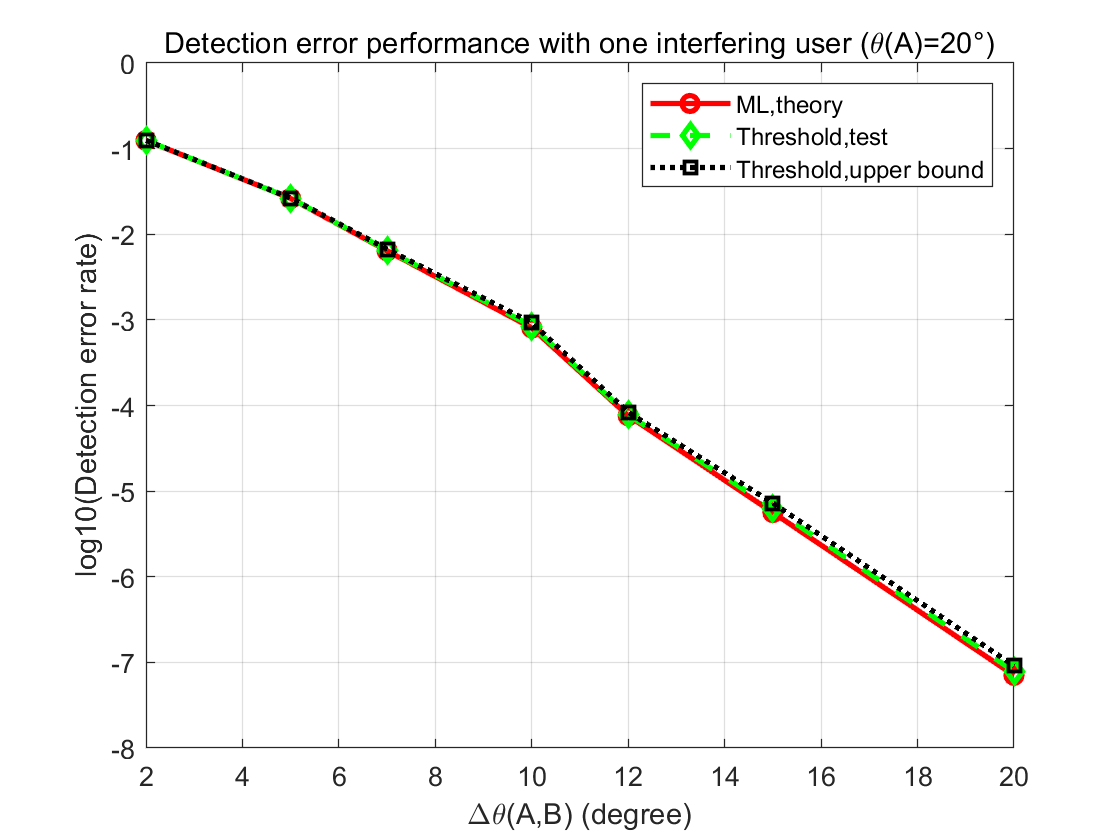}}
	\caption{The detection error rate with one interfering user.}
	\label{interf_1} 
\end{figure}

Now we turn to three-users system (desired user $A$, interfering users $B_1$ and $B_2$) and four-users system (desired user $A$, interfering users $B_1$, $B_2$ and $B_3$). Set the distance between each user and the receiver to $100m$. We aim to test the detection error rate with respect to different users position. For the three-user system, set $\Delta\theta(A,B_1)=\Delta\theta(B_1,B_2)=\Delta_2$ and change the value of $\Delta_2$; while for the three-user system, set $\Delta\theta(A,B_1)=\Delta\theta(B_1,B_2)=\Delta\theta(B_2,B_3)=\Delta_3$ and change the value of $\Delta_3$. Fig.~\ref{multiuser_detect} shows the detection error rate with 1,2 and 3 interfering users. Obviously, a larger amount of interfering users increases the detection error rate. It is shown that $P_e^{th}$ is still an excellent closed-form approximation of $P_e^{ML}$ and the successive elimination method is still feasible.   

\begin{figure}[htbp]
	\centering
	\includegraphics[width=4.5in]{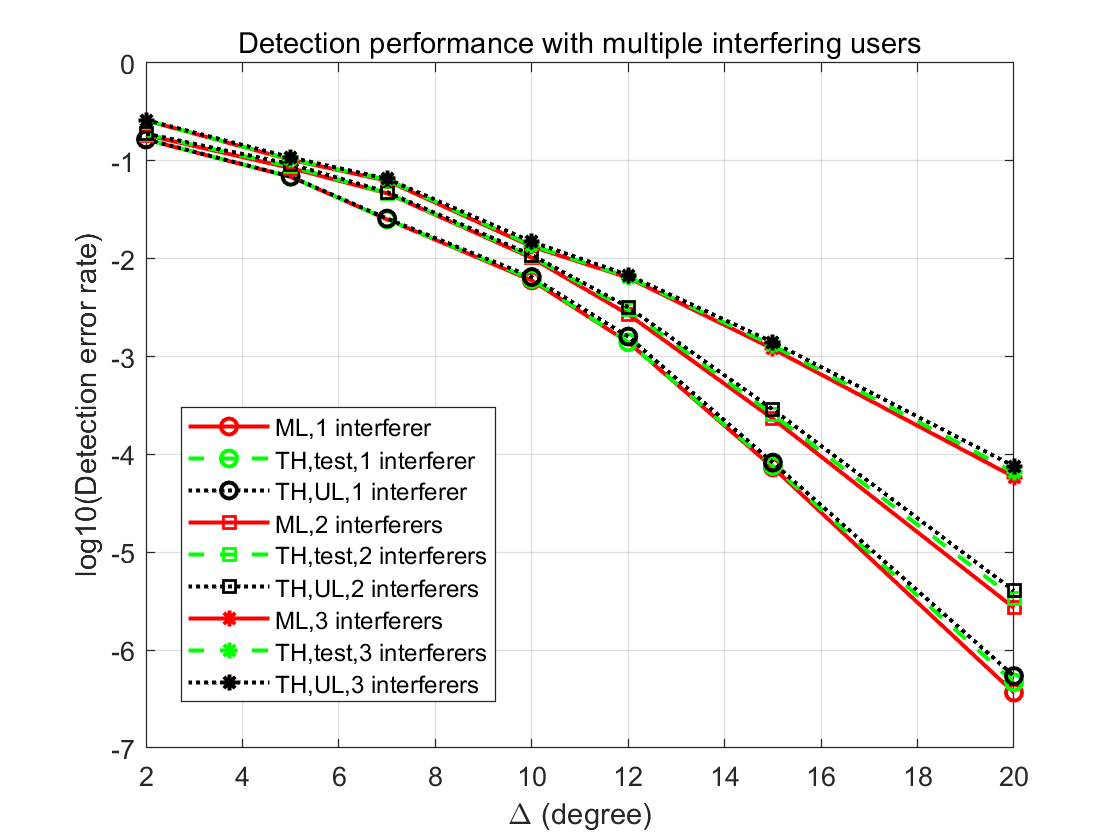}
	\caption{The detection error rate with multiple interfering users.} \label{multiuser_detect}
\end{figure}

Finally, we consider three-user system, and compare the whole execution time of threshold detection ($t_{TH}$) with that of ML detection ($t_{ML}$) via simulations. Fig.~\ref{ml_th_prop} depicts ratio $\frac{t_{ML}}{t_{TH}}$ under different channel conditions, showing that the proposed successive elimination method is nearly 20 times faster than ML detection. 

\begin{figure}[htbp]
	\centering
	\includegraphics[width=4in]{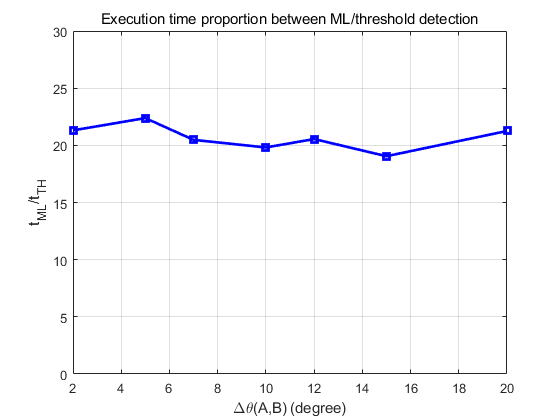}
	\caption{The execution time ratio of ML to the proposed successive elimination detection.} \label{ml_th_prop}
\end{figure}

\section{conclusion}
We have addressed the channel estimation and signal detection via exploiting the space division potential of ultraviolet scattering communication. We have optimized the pilot matrix to reduce the MSE of the LS channel estimation. In order to characterize the space division in MAC system, we have adopted Gaussian approximation on the Poisson weighted sum, and analyzed the separability of the two users by the optimal threshold detection rule. We have also addressed the signal detection in the case of multiple users based on the previous threshold detection rule, and proposed a successive elimination method with a tractable closed-form expression of the upper bound on the detection error probability. The computational complexity of the proposed approach is around twenty times lower than that of the ML detection.

\begin{appendices}
	\section{ Proof of Theorem \ref{f_mse_deep_calc} } \label{appendix_A}
	Recall that the unbiased estimation mean-square-error of m-th PMT with respect to pilot matrix $\boldsymbol{X}_{K \times L}$ is shown as follows,
	\begin{align} \label{unbiased_H_MSE_app}
	 MSE_m(\boldsymbol{X})=Tr((\boldsymbol{X}\boldsymbol{X}^T)^{-1}\boldsymbol{X}diag(\boldsymbol{X}^T\boldsymbol{\lambda_{sr}}(m))\boldsymbol{X}^T(\boldsymbol{X}\boldsymbol{X}^T)^{-1})+\boldsymbol{\lambda_n}(m)Tr((\boldsymbol{X}\boldsymbol{X}^T)^{-1}),
	\end{align}
	where $\boldsymbol{\lambda_n}(m)$ is the Poisson noise intensity of the m-th PMT, $K$ is total users number, and $\boldsymbol{\lambda_{sr}}(m)$ is the estimated channel parameters vector. 
	Define $K \times K$ identity matrix as $\boldsymbol{I}$ and $K \times K$ full-one matrix as $\boldsymbol{E}$. According to Eq. (\ref{abc_def}), we have
	\begin{align} \label{xxt_expression}
	\boldsymbol{X}\boldsymbol{X}^T=L[(a-b)\boldsymbol{I}+b\boldsymbol{E}].
	\end{align}
	 Recalling that $\boldsymbol{X}=(\boldsymbol{x}_1,\boldsymbol{x}_2,...,\boldsymbol{x}_L)$, and noting that each element of matrix $\boldsymbol{X}$ is either 0 or 1, we have
	\begin{align} \label{norm_1_property}
	\sum_{l=1}^{L} \Vert \boldsymbol{x}_l \Vert_1^2&=\sum_{l=1}^{L} \bigg(\sum_{k=1}^{K} x_{l,k}+2\sum_{1 \leq k_1 \textless k_2 \leq K} x_{l,k_1}x_{l,k_2}\bigg)
	=aKL+bK(K-1)L. \notag \\
	\sum_{l=1}^{L} \Vert \boldsymbol{x}_l \Vert_1^3&=\sum_{l=1}^{L} \bigg(\sum_{k=1}^{K} x_{l,k}+6\sum_{1 \leq k_1 \textless k_2 \leq K} x_{l,k_1}x_{l,k_2}+6\sum_{1 \leq k_1 \textless k_2 \textless k_3 \leq K} x_{l,k_1}x_{l,k_2}x_{l,k_3} \bigg) \notag \\
	&=aKL+3bK(K-1)L+cK(K-1)(K-2)L.
	\end{align}
	Using Matrix Inverse Lemma, we have
	\begin{align} \label{xxt_inv_expression}
	(\boldsymbol{X}\boldsymbol{X}^T)^{-1}=\frac{1}{L(a-b)}[\boldsymbol{I}-\frac{b}{a+(K-1)b}\boldsymbol{E}].
	\end{align}
	Based on the properties of matrix trace, we have
	\begin{align} \label{trace_order_change}
	Tr((\boldsymbol{X}\boldsymbol{X}^T)^{-1}\boldsymbol{X}diag(\boldsymbol{X}^T\boldsymbol{\lambda_{sr}}(m))\boldsymbol{X}^T(\boldsymbol{X}\boldsymbol{X}^T)^{-1})
	=Tr(\boldsymbol{X}^T(\boldsymbol{X}\boldsymbol{X}^T)^{-1}(\boldsymbol{X}\boldsymbol{X}^T)^{-1}\boldsymbol{X}diag(\boldsymbol{X}^T\boldsymbol{\lambda_{sr}}(m))).
	\end{align}
	Furthermore, we have
	\begin{align} \label{Q_def}
	\boldsymbol{Q^{'}}=(\boldsymbol{X}\boldsymbol{X}^T)^{-1}(\boldsymbol{X}\boldsymbol{X}^T)^{-1}
	=\frac{1}{L^2(a-b)^2}[\boldsymbol{I}-\frac{2ab-(2-K)b^2}{(a+(K-1)b)^2}\boldsymbol{E}]=\frac{\boldsymbol{Q}}{L^2(a-b)^2}. 
	\end{align}
	Then, we have
	\begin{align} \label{mse_tractable_expression}
	MSE_m(\boldsymbol{X})=Tr((\boldsymbol{X}^T\boldsymbol{Q^{'}}\boldsymbol{X})diag(\boldsymbol{X}^T\boldsymbol{\lambda_{sr}}(m)))+\boldsymbol{\lambda_n}(m) \ Tr((\boldsymbol{X}\boldsymbol{X}^T)^{-1}).
	\end{align}
	Moreover, we have
	\begin{align} \label{trace_first_part}
	Tr((\boldsymbol{X}\boldsymbol{X}^T)^{-1})=\frac{K}{L(a-b)} \ (1-\frac{b}{a+(K-1)b}).
	\end{align} 
	Considering term $Tr((\boldsymbol{X}^T\boldsymbol{Q^{'}}\boldsymbol{X})diag(\boldsymbol{X}^T\boldsymbol{\lambda_{sr}}(m)))$, recalling
	\begin{align} \label{x_q_h_clear_express}
	\boldsymbol{X}_{K \times L}=(\boldsymbol{x}_1,\boldsymbol{x}_2,...,\boldsymbol{x}_L), \ 
	\boldsymbol{\lambda_{sr}}(m)=(\boldsymbol{\lambda_s}(1,m),\boldsymbol{\lambda_s}(2,m),...,\boldsymbol{\lambda_s}(K,m))^T,
	\end{align} 
	we have
	\begin{align} \label{hard_part_calc_1}
	Tr((\boldsymbol{X}^T\boldsymbol{Q^{'}}\boldsymbol{X})diag(\boldsymbol{X}^T\boldsymbol{\lambda_{sr}}(m)))=\frac{1}{L^2(a-b)^2} \ \bigg(\sum_{l=1}^{L} \ \boldsymbol{x}_l^T\boldsymbol{Q}\boldsymbol{x}_l\boldsymbol{x}_l^T\bigg)\boldsymbol{\lambda_{sr}}(m).
	\end{align}
	According to Eq. (\ref{Q_def}), to simplify the calculation, denoting $\boldsymbol{Q}=\boldsymbol{I}+\psi \boldsymbol{E}$, we have that,
	\begin{align} \label{tr_2_calc_1}
	\sum_{l=1}^{L} \ \boldsymbol{x}_l^T\boldsymbol{Q}\boldsymbol{x}_l\boldsymbol{x}_l^T=&\sum_{l=1}^{L} \ \boldsymbol{x}_l^T (\boldsymbol{I}+\psi \boldsymbol{E})\boldsymbol{x}_l\boldsymbol{x}_l^T 
	=\sum_{l=1}^{L} \boldsymbol{x}_l^T \boldsymbol{I} \boldsymbol{x}_l \boldsymbol{x}_l^T + \psi \sum_{l=1}^{L} \boldsymbol{x}_l^T \boldsymbol{E} \boldsymbol{x}_l \boldsymbol{x}_l^T \notag \\
	=&\sum_{l=1}^{L} (\sum_{k=1}^{K} x^2_{k,l}) \boldsymbol{x}_l^T + \psi \sum_{l=1}^{L} (\sum_{k=1}^{K} x_{k,l})^2 \boldsymbol{x}_l^T.
	\end{align}
	Define
	\begin{align} \label{q_def}
	q(k_0)=\sum_{l=1}^{L} (\sum_{k=1}^{K} x^2_{k,l}) x_{k_0,l} + \psi \sum_{l=1}^{L} (\sum_{k=1}^{K} x_{k,l})^2 x_{k_0,l}, \ k_0=1,2,...,K.
	\end{align} 
	Due to user balanced structure for $\boldsymbol{X}$, we have 
	\begin{align}
	q(1)=q(2)=...=q(K) \triangleq q.
	\end{align} 
	To simplify the calculation of $q$, consider the summation of $\{q(j)\}$,
	\begin{align}
	Kq&=\sum_{k_0=1}^{K} q(k_0) =\sum_{l=1}^{L} (\sum_{k=1}^{K} x^2_{k,l}) (\sum_{k_0=1}^{K} x_{k_0,l}) + \psi \sum_{l=1}^{L} (\sum_{k=1}^{K} x_{k,l})^2 (\sum_{k_0=1}^{K} x_{k_0,l}) \notag \\
	&=\sum_{l=1}^{L} (\sum_{k=1}^{K} x_{k,l}) (\sum_{k_0=1}^{K} x_{k_0,l}) + \psi \sum_{l=1}^{L} (\sum_{k=1}^{K} x_{k,l})^3
	=\sum_{l=1}^{L} (\sum_{k=1}^{K} x_{k,l})^2 + \psi \sum_{l=1}^{L} (\sum_{k=1}^{K} x_{k,l})^3 \notag \\
	&=\sum_{l=1}^{L} \Vert \boldsymbol{x}_l \Vert_1^2 + \psi \sum_{l=1}^{L} \Vert \boldsymbol{x}_l \Vert_1^3.
	\end{align}
	Thus, we have
	\begin{align} \label{q_calc}
	q&=\frac{1}{K} (\sum_{l=1}^{L} \Vert \boldsymbol{x}_l \Vert_1^2 + \psi \sum_{l=1}^{L} \Vert \boldsymbol{x}_l \Vert_1^3) \notag \\
	&=aL+b(K-1)L+\psi(aL+3b(K-1)L+c(K-1)(K-2)L).
	\end{align}
	Denote $\boldsymbol{1}$ as a $K \times 1$ full-one vector, note that,
	\begin{align} \label{tr_2_calc_2}
	\bigg(\sum_{l=1}^{L} \ \boldsymbol{x}_l^T\boldsymbol{Q}\boldsymbol{x}_l\boldsymbol{x}_l^T\bigg)\boldsymbol{\lambda_{sr}}(m)=q \cdot \boldsymbol{1} \cdot \boldsymbol{\lambda_{sr}}(m)=q \Vert \boldsymbol{\lambda_{sr}}(m) \Vert_1.
	\end{align}
	Based on Eq. (\ref{mse_tractable_expression}), Eq. (\ref{trace_first_part}), Eq. (\ref{hard_part_calc_1}), Eq. (\ref{q_calc}) and Eq. (\ref{tr_2_calc_2}), noting that $\psi=-\frac{2ab-(2-K)b^2}{(a+(K-1)b)^2}$, we have
	\begin{align} \label{mse_x_final}
	&\mathcal{F}_{m}(a,b,c)=\frac{\boldsymbol{\lambda_n}(m) K(1-\frac{b}{a+(K-1)b})}{L(a-b)} \notag \\
	&+\frac{\Vert \boldsymbol{\lambda_{sr}}(m) \Vert_1}{L(a-b)^2}\bigg( a+(K-1)b+\frac{(2-K)b^2-2ab}{(a+(K-1)b)^2} (a+3b(K-1)+c(K-1)(K-2))\bigg).
	\end{align}
	
\end{appendices}

\begin{footnotesize}
	\bibliography{chan_esti_sig_detect_SDMA}
	\bibliographystyle{IEEEtran}
\end{footnotesize}

\end{document}